\documentclass[a4paper,UKenglish,cleveref,thm-restate]{lipics-v2021}
\hideLIPIcs

\usepackage{amsmath,amssymb}
\usepackage{enumitem}
\usepackage{lineno}
\usepackage{tikz,cite}
\usepackage{graphicx}
\usepackage{float}
\usepackage{ifthen}
\usepackage{todonotes}
\usepackage{xcolor}[dvipsnames]

\title{On a tree-based variant of bandwidth and forbidding simple topological minors}

\author{Hugo Jacob}{LIRMM, Université de Montpellier, CNRS, Montpellier,
France}{hugo.jacob@lirmm.fr}{https://orcid.org/0000-0003-1350-3240}{Supported
by the ANR project GODASse ANR-24-CE48-4377.}

\author{William Lochet}{LIRMM, Université de Montpellier, CNRS, Montpellier,
France}{william.lochet@gmail.com}{https://orcid.org/0000-0002-8711-1170}{Supported
by the ANR project CADO ANR-24-CE48-3758-01.}

\author{Christophe Paul}{LIRMM, Université de Montpellier, CNRS, Montpellier,
France}{paul@lirmm.fr}{https://orcid.org/0000-0001-6519-975X}{Supported by the
ANR project GODASse ANR-24-CE48-4377.}

\funding{}

\authorrunning{H. Jacob, W. Lochet, and C. Paul} 

\Copyright{Hugo Jacob, William Lochet, and Christophe Paul}

\keywords{Tree decompositions, bandwidth, topological minors, approximation algorithms} 

\ccsdesc[500]{Theory of computation~Graph algorithms analysis}

\acknowledgements{The first author wishes to thank Hans L. Bodlaender for
discussions on the hardness of computing treebandwidth, Dimitrios M.
Thilikos for asking good questions related to the content of this paper, and
Evangelos Protopapas for initial discussions on the topic.} 

\nolinenumbers 

\theoremstyle{definition}
\newtheorem{question}{Question}
\declaretheorem[name={Open question}]{opquestion}

\DeclareMathOperator{\bw}{\mathsf{bw}}
\DeclareMathOperator{\pw}{\mathsf{pw}}
\DeclareMathOperator{\tw}{\mathsf{tw}}
\DeclareMathOperator{\dtw}{\mathsf{dtw}}
\DeclareMathOperator{\tbw}{\mathsf{tbw}}
\DeclareMathOperator{\tpw}{\mathsf{tpw}}

\DeclareMathOperator{\td}{\mathsf{td}}
\DeclareMathOperator{\ts}{\mathsf{ts}}
\DeclareMathOperator{\vc}{\mathsf{vc}}
\DeclareMathOperator{\fvs}{\mathsf{fvs}}
\DeclareMathOperator{\etw}{\mathsf{etw}}
\DeclareMathOperator{\otw}{\mathsf{otw}}
\DeclareMathOperator{\p}{\mathsf{p}}

\newcommand{\Series}{\mathsf{S}}
\newcommand{\Parallel}{\mathsf{P}}
\newcommand{\Q}{\mathsf{Q}}

\newcommand{\Rigid}{\mathsf{R}}
\newcommand{\tb}{\texttt{tb}}
\newcommand{\te}{\texttt{te}}

\newcommand{\defparaproblem}[4]{
 \vspace{2mm}
\noindent\fbox{
 \begin{minipage}{0.96\textwidth}
 \begin{tabular*}{\textwidth}{@{\extracolsep{\fill}}lr} #1 & \\ \end{tabular*}
 {\textbf{Input:}} #2 \\
 {\textbf{Parameter:}} #3 \\
 {\textbf{Question:}} #4
 \end{minipage}
 }
 \vspace{2mm}
}

\newenvironment{subproof}[1][\proofname]{%
  \begin{proof}[#1]%
}{%
  \end{proof}%
}

\ccsdesc[500]{Theory of computation~Graph algorithms analysis}

\begin{document}

\maketitle

\begin{abstract}
	We obtain structure theorems for graphs excluding a fan (a path with a
	universal vertex) or a dipole ($K_{2,k}$) as a topological minor. The
	corresponding decompositions can be computed in FPT linear time. This is motivated by the study of a graph parameter we call treebandwidth
	which extends the graph parameter bandwidth by replacing the linear layout by a
	rooted tree such that neighbours in the graph are in ancestor-descendant relation in
	the tree.
	
	We deduce an approximation algorithm for treebandwidth running in FPT linear
	time from our structure theorems. We complement this result with a precise
	characterisation of the parameterised complexity of computing the parameter exactly.
\end{abstract}

\section{Introduction}

Bandwidth is a classical graph parameter that originates from the study of sparse matrices
in the 1960s (see e.g. \cite{CuthillM69}). Given a matrix $M$, the objective was to find a
permutation matrix $P$ such that $PMP^T$ has its nonzero entries next to the diagonal. In
the graph setting, this amounts to computing a linear ordering, called \emph{layout}, of
the vertices of a given graph such that adjacent vertices are kept close to each other.
More formally, given a graph $G=(V,E)$, a layout is a bijective mapping
$\sigma: V \to \{1,\dots,n\}$. Then the \emph{bandwidth} of a layout of graph $G$
is defined as $$\bw(\sigma)=\max_{xy\in E} |\sigma(x)-\sigma(y)|,$$ and the
bandwidth of $G,$ $\bw(G),$ is the minimum of $\bw(\sigma)$ over all layouts $\sigma$ of $G$.

Our initial motivation in this paper is to study the graph parameter
corresponding to bandwidth where instead of mapping vertices of the graph to a
layout, we map them to a rooted tree. A \emph{tree-layout} of $G=(V,E)$ is a
rooted tree $T$ whose nodes are the vertices of $V$, and such that, for every edge
$xy\in E$, $x$ is an ancestor of $y$ or vice-versa. The bandwidth of $T$ is
then the maximum distance in $T$ between pairs of neighbours in $G$. We call
\emph{treebandwidth} of $G$, the minimum bandwidth over tree-layouts of $G$,
and denote it by $\tbw(G)$. A definition via tree-layouts can also be obtained
for the usual graph parameters treewidth and treedepth by minimising the size
of open neighbourhoods of subtrees or the depth, respectively.

To motivate this notion, we provide two equivalent definitions via two other
common ways to define graph parameters. On one hand, similarly to many graph
parameters related to treewidth, we can define it via the maximum clique in a
minimal completion to a subclass of chordal graphs. In this case, the class of
proper chordal graphs recently introduced in \cite{properchordal}. On the other
hand, we can define it via a variant of graph search games (also called cops
and robber games).

\begin{proposition}\label{prop:equiv-intro}
Given a fixed graph $G$, the following statements are equivalent:
\begin{enumerate}[label=(\roman*)]
\item $\tbw(G) \leq k$;
\item $\omega(H)-1 \leq k$ for a proper chordal supergraph $H$
of $G$; and
\item there is a monotone search strategy to capture a visible fugitive
of infinite speed in which searchers are placed one at a time and each vertex of $G$ is
occupied by a searcher during at most $k+1$ time steps.
\end{enumerate}
\end{proposition}

The question of embedding graphs in trees with bounded \emph{stretch} (distance
between neighbours in the embedding) dates back to Bienstock
\cite{Bienstock90OnEmbedding}. The fact that any graph can be embedded in a
star (tree of diameter $2$) lead Bienstock to restrict the attention to
bounded degree trees. Surprisingly, graphs that admit such an embedding of
bounded stretch are exactly graphs with bounded treewidth and bounded maximum
degree \cite{DingO95Some}. Another way of avoiding the trivial embedding in a
star while still considering trees of unbounded degree is to embed the graph in an
arborescence with the constraint that neighbours should be linked by a directed
path. In this case, the arborescence is exactly a tree-layout whose bandwidth
corresponds to the stretch of the embedding.

The central questions around which this paper is organised are the following:

\begin{question}\label{Question:struct}
What is the structure of graphs of bounded treebandwidth?
Can we efficiently recognise such graphs?
\end{question}

We also provide a simple algorithmic application of the parameter
treebandwidth via $p$-centered colourings, and a presentation of how this parameter
relates to other similar parameters.

\paragraph*{Main contributions}

We characterise graphs of bounded treebandwidth by two parametric forbidden topological
minors, namely the fan\footnote{A dominated path} and the wall. Note that excluding a wall as a topological minor bounds the treewidth.

\begin{theorem}\label{thm:tbw-obs}
Graphs of bounded treebandwidth are exactly graphs which exclude large walls and large fans as topological minors.
\end{theorem}

This result is obtained as a corollary of a structure theorem for graphs excluding a $k$-fan as a topological minor.

\begin{theorem}\label{thm:folded-fan-intro}
If $G$ excludes $F_k$ as a topological minor, then there exists a tree
decomposition $(T,\beta)$ of $G$ such that adhesion size is at most $a$, each vertex
$v$ has at most $b$ neighbours in the torso at each bag, and the bags containing $v$
induce a subtree of $T$ of diameter at most $c$.

Such a decomposition can be computed in time $f(k)\cdot n$.

Conversely, if such a decomposition exists, $G$ excludes $F_{k'}$ as a
topological minor for some $k'$ function of $a,b,c$.
\end{theorem}

Our structure theorem follows from the identification of precise invariants
that any `sufficiently good' tree decomposition satisfies, and a technique to
modify tree decompositions in order to reduce the span of vertices. In
particular, it is easy to incorporate the assumption that the graph has bounded
treewidth. Since the structure theorem is algorithmic, we deduce an FPT linear
algorithm to approximate treebandwidth.

\begin{theorem}\label{thm:tbw-algo-intro}
	There is an algorithm which, given a graph $G$ and an integer $k$, in time
	$f(k)\cdot n,$ either computes a tree-layout of bandwidth at most $g(k)$, or
	determines that $\tbw(G)>k$.
\end{theorem}

We complement this result with a hardness proof for the parameterised problem
\textsc{Treebandwidth} of computing treebandwidth. This excludes the possibility of an
FPT algorithm to compute treebandwidth exactly under reasonable assumptions.

\begin{theorem}
\textsc{Treebandwidth} is in XALP, and is XALP-hard even when parameterised by
treewidth and maximum degree.
\end{theorem}

Following the strategy we used for fans, we also establish a structure theorem
for graphs excluding a $k$-dipole\footnote{The multigraph consisting of $2$
vertices with $k$ edges between them} as a topological minor. Before our two
structure theorems, the only structure theorem for graphs excluding a single
topological minor of unbounded degree is the structure theorem for graphs
excluding a clique as a topological minor of Grohe and Marx
\cite{HadwigerGroheMarx,HadwigerDvorak}. On the other hand, many studied graph
parameters can be characterised by a finite list of topological minors. We
review them to provide an exposition of how treebandwidth compares to other
parameters in the literature.

\paragraph*{Related literature}

\subparagraph{Bandwidth and complexity.}
We begin with a brief presentation of the literature on the parameter bandwidth.
The problem \textsc{Bandwidth} of computing the bandwidth of a graph is particularly hard.
It is NP-hard even when restricted to caterpillars of bounded hair length or bounded
degree \cite{BandwidthNPhard}. The hardness of \textsc{Bandwidth} on caterpillars extends
to the parameterised setting: as first claimed by Bodlaender, Fellows, and
Hallett~\cite{BandwidthWhierarchy}, the parameterised bandwidth problem is hard for W[$t$]
for all $t$ (see \cite{BandwidthXNLP} for a proof). More precisely, \textsc{Bandwidth} is
complete for the class XNLP of parameterised problems\footnote{XNLP is defined in
Section~\ref{sec:hardness}.}, as well as many other problems parameterised by the bandwidth
of the input graph, see \cite{XNLP0,XNLP1journal,XNLP2} and \cite{BandwidthPreXNLP} in the
classical setting. Approximating the bandwidth of a graph is APX-hard, even on
caterpillars~\cite{Unger98TheComplexity}, but there is an FPT algorithm to approximate
bandwidth on trees \cite{BandwidthTrees}. For more results on the exact or approximate
computation of bandwidth, see~\cite{Cygan11Exact}.

\subparagraph{Structure of bounded bandwidth graphs.}
Despite the early interest in bandwidth, the combinatorial structure of graphs
of bounded bandwidth is not completely understood. From the definition of bandwidth, it is
easy to observe that bounded bandwidth implies bounded pathwidth and bounded degree. 
The degree lower bound can be strengthened in terms of the so-called \emph{local density} 
of a graph (see \cite{ChinnCDH82TheBandwidth,ChungSeymourBandwidth}), but this still 
does not yield a characterisation of bandwidth. Indeed, some caterpillars have bounded 
local density but unbounded bandwidth \cite{ChungSeymourBandwidth}. The parameterised
approximation algorithm on trees of \cite{BandwidthTrees} is essentially based on the fact
that, for trees, the obstructions are those listed above. The following question of Chung
and Seymour still stands.

\begin{opquestion}
	Is it true that the bandwidth of a graph is bounded by a function of the maximum
	bandwidth over its tree subgraphs?
\end{opquestion}

\subparagraph{Topological bandwidth.}
One of the main obstacles for a structural characterisation of graphs of bounded
bandwidth is that bandwidth does not seem to behave well with graph relations
except for the subgraph relation. This lack of characterisation may explain
why computing or approximating the bandwidth of a graph is hard and has motivated the
consideration of topological bandwidth as an alternative graph parameter. The
\emph{topological bandwidth} of a graph $G$ is the minimum bandwidth over all the
subdivisions of $G$. A graph $H$ is a \emph{topological minor} of $G$ if $G$ contains a
subdivision of $H$ as a subgraph. By definition, topological bandwidth is monotone for the
topological minor relation. This makes it possible to characterise graphs of bounded
topological bandwidth as graphs excluding large binary trees and large stars as
topological minors~\cite{ChungSeymourBandwidth}. 
Excluding large complete binary trees and stars as topological minors
respectively bounds the pathwidth and the degree of a graph (the two trivial bandwidth
lower bounds discussed above). These obstructions imply an FPT approximation algorithm for
topological bandwidth (e.g. via the computation of the equivalent parameter cutwidth
\cite{cutwidth}).

\subparagraph{Tree-like generalisations of bandwidth.} It is relatively natural to consider
tree-like generalisations of bandwidth or cutwidth. They correspond exactly to optimising
the stretch and congestion of an embedding on a path as noted by Bienstock
\cite{Bienstock90OnEmbedding}, who then considers stretch and congestion of embeddings in
bounded degree trees. Surprisingly, in this context, it turns out that graphs admitting an
embedding of bounded stretch are exactly graphs admitting an embedding of bounded
congestion, which are exactly graphs of bounded degree and bounded treewidth
\cite{Bienstock90OnEmbedding,DingO95Some}. This is captured by the parameter \emph{domino
treewidth} which is the minimum width of a tree decomposition\footnote{Tree decompositions
are defined in \cref{subsec:treedec}} such that each vertex is contained in at most two
bags (called \emph{domino tree decomposition}), as domino treewidth is also known to be
bounded exactly on graphs of bounded degree and bounded treewidth
\cite{DominoTw1,DominoTw2}. In \cite{FominHT05GraphSearching}, the parameter
\emph{treespan} is introduced as a generalisation of bandwidth following its
characterisation by a
graph searching game which we describe later. We will show that this parameter is also
bounded exactly on graphs of bounded degree and bounded treewidth (\cref{thm:eq-treespan-dtw}).
Instead of a tree decomposition, one may consider a \emph{tree-partition}
$(T,\mathcal{P})$ of a graph $G$, where $\mathcal{P}$ is a partition of the
vertices of $G$ whose parts are indexed by nodes of a tree $T$ in a way that, for
every edge $xy$ of $G$, either $x$ and $y$ are in the same part of $\mathcal{P}$, or they
are contained in parts indexed by adjacent nodes of $T$. The width of the tree-partition
is the size of the largest part of $\mathcal{P}$ and the \emph{tree-partition-width} of
$G$, denoted $\tpw(G)$, is the minimum width of a tree-partition of $G$. The notion was
introduced by Halin and Seese \cite{Halin91,Seese85}. This can also be expressed in terms
of a graph morphism to a tree (with self-loops) where the width is the maximum number of
preimages of a node of the tree. Observe that the bandwidth of a graph is linearly related
to its path-partition-width (see \cite{DujmovicSW07}). Thereby tree-partition-width is an
accurate tree-like analogue of bandwidth, it captures graphs of bounded treewidth and
bounded degree (there is always a tree-partition of width $O(\Delta(G)\tw(G)),$ see
\cite{tpwDegree,tpwWood}), but also some graphs of unbounded degree. 

\subparagraph{Metric flavoured variants.} The following notions seem quite similar but the
embedding is more constrained by the graph $G$, leading to very different classes of
graphs. Following the idea of keeping close to each other adjacent vertices,
the \emph{tree-length}~\cite{DourisboureG07TreeDecomposition} measures
the maximum diameter of the subgraph induced by a bag of a tree decomposition (this is
related to quasi-isometry to trees see \cite{BergerS24BoundedDiameter}). 
In~\cite{CarrollGM06Embedding} and \cite{BartalCMN13Bandwidth}, they consider the stretch in
a tree $T$ which must be a DFS-tree of $G$. The resulting parameter, which they called
\emph{tree-bandwidth}, admits another definition in terms of tree-partitions with
additional constraints. Considering the stretch in spanning trees leads to the notion of
$t$-spanner introduced in \cite{CaiC95}.

\subparagraph{Treebandwidth.} 
We now move to a presentation of the parameter treebandwidth we consider in this paper.
The parameter was first considered in \cite{spanheight} where it was called
\emph{spanheight}. We observe that the \emph{treedepth} of a graph
$G$~\cite{NesetrilO12Sparsity}, denoted $\td(G)$, is also defined using tree-layouts.
Indeed, $\td(G)$ is the minimum depth of a tree-layout (called elimination forest in this
context), implying that $\tbw(G)\leq \td(G)$. Moreover, in the same way that the bandwidth
of a graph $G$ is a lower bound to the \emph{pathwidth} $\pw(G)$, the treebandwidth of $G$
is a lower bound to the \emph{treewidth} $\tw(G)$. We believe that treebandwidth is a
natural parameter that deserves a thorough study as it gathers several important features.
As mentioned earlier, treebandwidth admits several equivalent definitions, which we now
compare to similar definitions of known parameters.

It is known that treedepth, bandwidth, pathwidth and treewidth of a graph $G$ are
respectively the maximum clique number $\omega(H)$ (minus one) minimised over completions
$H$ of $G$ into a trivially perfect graph, a proper interval graph, an interval graph and
a chordal graph (see Figure~\ref{fig:chordal-param}). The recently introduced class of
proper chordal graphs~\cite{properchordal} is designed as a tree-like analogue of proper
interval graphs: a graph $G$ is proper chordal if it admits a tree-layout $T$ such that
every root-to-leaf path of $T$ induces a proper interval subgraph of $G$. Treebandwidth
is equal to $\omega(H)-1$ for some completion $H$ of $G$ into a proper chordal graph.

The inclusion relations between these graph classes witness the relation discussed above
between these parameters.

\begin{figure}[ht]
\begin{center}
\bigskip
\begin{tikzpicture}[thick,scale=0.7]
\tikzstyle{sommet}=[circle, draw, fill=black, inner sep=0pt, minimum width=4pt]

\begin{scope}[xshift=-3cm,yshift=0cm]
\draw[->] (-5,1.3) -- (-5,4.7) node[midway,above,sloped] {$\subseteq$};

\node[] (C) at (1.2,5) {} ;
\draw[] (C) node[] {Chordal};
\draw[] (0.2,4.5) rectangle (2.2,5.5);

\node[] (I) at (-2,3.8) {} ;
\draw[] (I) node[] {Interval};
\draw[] (-3,3.3) rectangle (-1,4.3);

\node[] (PC) at (2,3) {} ;
\draw[] (PC) node[] {Proper Chordal};
\draw[] (0.2,2.5) rectangle (3.8,3.5);

\node[] (PI) at (-2.8,1.8) {} ;
\draw[] (PI) node[] {Proper Interval};
\draw[] (-4.6,1.3) rectangle (-1,2.3);

\node[] (TP) at (2.1,1) {} ;
\draw[] (TP) node[] {Trivially Perfect};
\draw[] (0.2,0.5) rectangle (4,1.5);

\draw[<-] (1.2,4.5) -- (1.2,3.5);
\draw[<-] (1.2,2.5) -- (1.2,1.5);
\draw[<-] (-2,3.3) -- (-2,2.3);
\draw[->] (-1,2) -- (0.2,3);
\draw[->] (-1,4) -- (0.2,5);
\draw[<-] (-1,3.6) -- (0.2,1);
\end{scope}

\begin{scope}[xshift=5cm,yshift=0cm]

\draw[dashed,->] (2,1.3) -- (2,4.7) node[midway,below,sloped] {$\geq$};

\node[] (tw) at (1,5) {} ;
\draw[] (tw) node[] {$\tw$};

\node[] (pw) at (-1,3.8) {} ;
\draw[] (pw) node[] {$\pw$};

\node[] (tbw) at (1,3) {} ;
\draw[] (tbw) node[] {$\tbw$};

\node[] (bw) at (-1,1.8) {} ;
\draw[] (bw) node[] {$\bw$};

\node[] (td) at (1,1) {} ;
\draw[] (td) node[] {$\td$};

\draw[dashed,<-] (1,4.7) -- (1,3.3);
\draw[dashed,<-] (1,2.7) -- (1,1.3);
\draw[dashed,<-] (-1,3.5) -- (-1,2.1);
\draw[dashed,->] (-0.8,2.1) -- (0.8,2.7);
\draw[dashed,->] (-0.8,4.1) -- (0.8,4.7);
\draw[dashed,<-] (-0.8,3.5) -- (0.8,1.3);
\end{scope}

\begin{scope}
\draw[dotted,<->] (-0.7,5) -- (5.6,5);
\draw[dotted,<->] (-3.9,3.8) -- (3.6,3.8);
\draw[dotted,<->] (0.9,3) -- (5.4,3);
\draw[dotted,<->] (-3.9,1.8) -- (3.6,1.8);
\draw[dotted,<->] (1.1,1) -- (5.6,1);
\end{scope}

\end{tikzpicture}
\end{center}
\caption{The hierarchy of subclasses of chordal graphs (on the left) and the corresponding hierarchy of graph parameters (on the right)
obtained by taking the maximum clique minimised over completions into the class.}
\label{fig:chordal-param}
\end{figure}
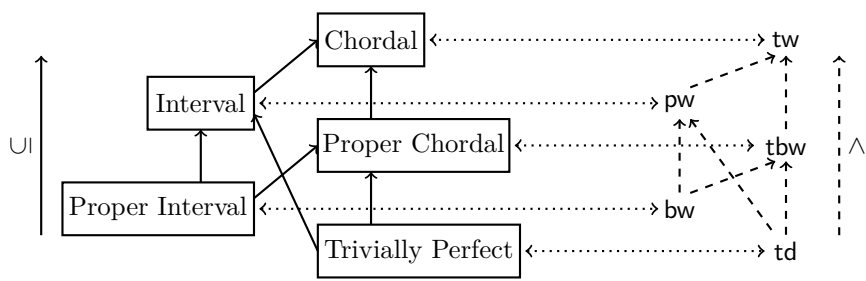

A second alternative definition of treebandwidth is obtained by considering monotone
graph searching games where a team of searchers aims at capturing a fugitive occupying
vertices and moving at infinite speed. The fugitive can be lazy (she moves only when the
searchers move to her position) or agile (she can move whenever she wants to), visible or
invisible to the searchers. It is known that for a graph $G$, the minimum number of
searchers to capture a lazy and invisible (or agile and visible) fugitive is $\tw(G)$,
while it is equal to $\pw(G)$ when the fugitive is agile and invisible
\cite{Searching1986}. Instead of minimising the searchers' number, let us consider
minimising the \emph{occupation time} defined as the maximum length of time a vertex
is occupied by a cop during a monotone\footnote{A search strategy is monotone if it
guarantees that a vertex occupied by a cop can never host the fugitive anymore. There is
no recontaminating.} search strategy. If the fugitive is agile and invisible, the
occupation time is equal to $\bw(G)$~\cite{RosenbergSudborough}. When the fugitive is lazy
and invisible, it is equal to a parameter called
\emph{treespan}~\cite{FominHT05GraphSearching} (see Figure~\ref{fig:search-param}).

\begin{figure}[h]
\centering
\begin{tabular}{|c|c|c|}
\hline
 & Cop number & Occupation time \\
\hline
Agile invisible & Pathwidth & Bandwidth \\
\hline
Lazy invisible & Treewidth & Treespan \\
\hline
Agile visible & Treewidth & Tree-bandwidth  \\
\hline
\end{tabular}
\caption{Table of the parameters corresponding to several variants of searching games.}
\label{fig:search-param}
\end{figure}
 
Interestingly, to our knowledge, occupation time (and thereby treebandwidth and treespan)
is the only known way to distinguish the cost of a monotone search against an agile
and visible fugitive from a monotone search against a lazy and invisible fugitive.
 
Finally, let us state an important feature of treebandwidth:
unlike bandwidth, treebandwidth behaves well with respect to the topological minor
relation. We can then hope that, similarly to topological bandwidth
\cite{ChungSeymourBandwidth}, the parameter treebandwidth admits a simple characterisation
by topological minor obstructions.

\begin{lemma}
If a graph $H$ is a subdivision of a graph $G$, then $\tbw(H) \leq \tbw(G) + 1$.
\end{lemma}
\begin{proof}
Suppose that $H$ is obtained from $G$ by subdividing an edge $xy$ of $G$ into a path
$P=(x=x_1,x_2,\dots,x_\ell=x)$. Let $T$ be a tree-layout of $G$ and assume that $x$ is an
ancestor of $y$ in $T$. Then, in $T$, we attach to $y$ a path $P_{T}$ of length $\ell-2$
and map the internal vertices of $P$ to the node of $P_{T}$ as indicated in
Figure~\ref{fig:subdivision}. Clearly, this construction increases by at most one the
bandwidth of the original tree-layout $T$.
\end{proof}
\begin{figure}[h]
\centering
\includegraphics[scale=.9]{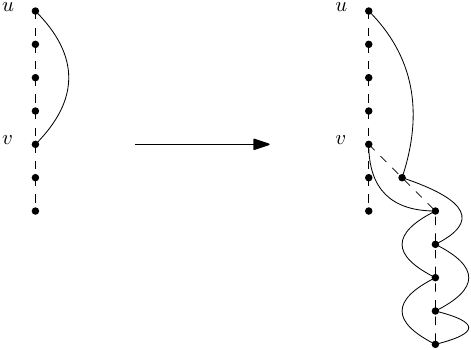}
\caption{Extending a tree-layout to a subdivision.}
\label{fig:subdivision}
\end{figure}

\subparagraph{Treebandwidth and tree-partition-width.} 
It is not immediate that treebandwidth and tree-partition-width define truly different
parameters. Indeed, it is easy to derive a tree-layout of bounded
bandwidth from a tree-partition of bounded width. By rooting the tree-partition
arbitrarily and replacing each bag by an arbitrary linear ordering of its vertices, one
derives $\tbw(G)\leq 2\cdot \tpw(G)$. However, some graphs of treebandwidth $2$ have
unbounded tree-partition-width: consider for example the $1$-subdivision of the
$k$-multiple of a $k$-star (see \cref{fig:obs-tpw}). 
Since tree-partition-width has been characterised asymptotically via forbidden topological
minors~\cite{obs-tpw} (\cref{fig:obs-tpw} depicts the set of topological minor
obstructions), and since treebandwidth also behaves well with respect to topological
minors, it is natural to search for a
characterisation of treebandwidth by topological minor obstructions. We observe that among the
obstructions to tree-partition-width, only the wall and the fan have unbounded
treebandwidth. This leads to the question whether there are other minimal obstructions.

\begin{figure}[h]
\centering
\includegraphics{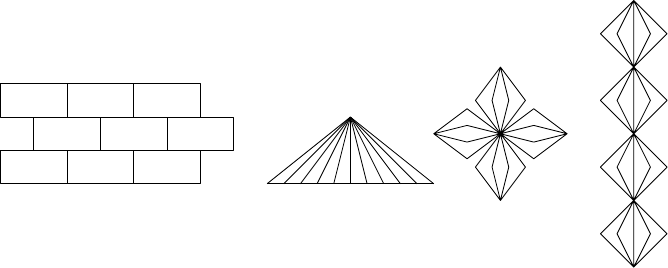}
\caption{The minimal obstructions to tree-partition-width: the $k$-wall, the $k$-fan,
the $k$-multiple of $k$-star, the $k$-multiple of $k$-path. Only the first two are obstructions to treebandwidth.}
\label{fig:obs-tpw}
\end{figure}

\subparagraph{Graphs excluding a $k$-fan.} As announced by \cref{thm:tbw-obs}, there are
no other obstructions. The proof relies on a
structure theorem for graphs excluding the $k$-fan as a topological minor. We show
that these graphs admit a tree decomposition with two special properties: one
limiting the degree of a vertex in each bag of the tree decomposition; the other limiting
what we could call the \emph{span} of the closed neighbourhood of any vertex over the
tree decomposition. Interestingly, while the structure of graphs excluding topological
minors has been investigated, known results are mostly about excluding cliques
\cite{HadwigerGroheMarx,HadwigerDvorak,ChunHungLiuThesis} and excluding bounded degree
graphs \cite{ForbBdDeg,ChunHungLiuThesis}.  
Keeping in mind that on subcubic graphs, the minor relation coincides with the topological
minor relation, the case of excluding a $k$-fan appears as a very interesting basic case to
investigate as it contains a unique vertex of unbounded degree. Motivated by this first
decomposition theorem, we push further our study on graphs excluding topological minors by
considering the case of \emph{dipoles} (graphs formed by two distinguished vertices, the
poles, and many internally vertex-disjoint paths between them).

\paragraph*{Organisation of the paper} 

In \cref{sec:spqr}, we consider the two simple cases of excluding a gem ($4$-fan), and
excluding $k$-fans in planar graphs, both via SPQR trees.
In \cref{sec:structure}, we present a first version of our structure theorems.
They are based on establishing the correct invariants to characterise an excluded fan
(resp. dipole) topological minor.
In \cref{sec:folding}, we present the folding technique and its application to our
structure theorems.
In \cref{sec:coro}, we present how to deduce approximation algorithms from applying
treewidth bounds to the folded structure theorems.
In \cref{sec:colouring}, we give a simple bound on $p$-centered colourings
and discuss its applications.
In \cref{sec:hardness}, we give the definitions of XNLP and XALP, and present the XALP-completeness result.
In \cref{sec:zoo}, we compare graph parameters related to treebandwidth from the
perspective of their topological minor obstructions.
In \cref{sec:def}, we prove the equivalent definitions announced in \cref{prop:equiv-intro}.

\section{Preliminaries}\label{sec:prelim}

We only consider simple graphs unless specified otherwise. Note, however, that any
multigraph can be turned into a simple graph by subdividing multiple edges and loops.

We say that two graph parameters $\p_1,\p_2$ are (polynomially) tied if there exist
(polynomial) functions $f_1,f_2$ such that, for every graph $G,$ $\p_1(G) \leq
f_1(\p_2(G))$ and $\p_2(G) \leq f_1(p_2(G))$. Since this defines an equivalence relation,
we may refer to two tied parameters as \emph{equivalent parameters}.

\subsection{Usual graph notations and terminology}

Given a graph $G$, and a vertex $x$, $N(x)$ denotes the set of neighbours of $x$ in $G$,
$\deg(x)$ the size of that set and $N[x]:= N(x) \cup \{u\}$. We denote the set of connected
components of a graph $G$ by $cc(G)$.
For any integer $k$, we denote by $P_k$ the path on $k$ vertices, $K_k$ the clique of size $k$, and $C_k$ the cycle of size $k$. 
Let $F_k$, the $k$-fan, denote the graph corresponding to $P_k$ with an additional
universal vertex. 
We call $k$-wheel the graph consisting of a cycle $C_k$ with an additional universal
vertex. We also denote by $K_{s,t}$ the complete bipartite graph where one size has size
$s$ and the other size $t$. The $k \times r$ grid corresponds to the graph on vertex set
$[k] \times [r]$ such that two vertices $(i,j)$ and $(i',j')$ are adjacent if
$|i-i'|+|j-j'|=1$. For some integer $k$, the $k$-wall, denoted by $W_k$, is the graph
obtained from the $2k \times k$ grid after removing all the edges of the form $\{
(x,y), (x,y+1) \}$ for $x+y$ odd and removing vertices of degree one. A $u$-$v$ path is a
path whose endpoints are $u$ and $v$. Given two subsets of vertices $U$ and
$V$, a $U$-$V$ path is a path with one endpoint in $U$ and one endpoint in $V$.
A \emph{$U$-$V$ separator} is a vertex set hitting all $U$-$V$ paths.

A \emph{minor model} of a graph $H$ in a graph $G$ is a mapping of each vertex $v$ of $H$ to
a connected subgraph $B_v$ of $G$, called \emph{branch set}, and a mapping of each edge
$uv$ of $H$ to an edge of $G$ incident to $B_u$ and $B_v$, such that the $B_v$ are
pairwise disjoint. $H$ is a minor of $G$ if there is a minor model of $H$ in $G$.
Given a subset of vertices $U$, a $U$-rooted minor model is a minor model whose branch
sets are all hit by $U$, and a $U$-rooted minor of $G$ is a graph that admits a $U$-rooted
minor model of $G$.

An \emph{immersion model} of a graph $H$ in a graph $G$, is an injective mapping of each
vertex $v$ of $H$ to a vertex $b_v$ of $G$, called \emph{branch vertex}, and a mapping of
each edge $uv$ of $H$ to a $u$-$v$ path $P_{uv}$ of $G$, such that the paths are pairwise
edge-disjoint.
An immersion model is a \emph{topological minor model} if the paths $P_{uv}$ are
pairwise internally vertex-disjoint. In this case, the model is a subgraph of $G$ which is a
subdivision of $H$.

A graph $H$ is a \emph{topological minor} (resp. \emph{immersion}) of a graph $G$ if there
is a model of $H$ in $G$. An immersion is \emph{strong} if the paths $P_{uv}$ of its model are
internally disjoint from the branch vertices, otherwise it is called \emph{weak}.
The wall, fan or dipole number of a graph $G$ is the largest integer $k$ such that $G$
contains, respectively, a $k$-wall, a $k$-fan and a $k$-dipole as a topological minor. 

We call tripod the graph consisting of the union of $3$ paths that share one of their
endpoints. Given three vertices $a,b,$ and $c$, an $(a,b,c)$-tripod is a tripod whose $3$ paths
have $a,b,c$ as their private endpoints.  Three vertices in a connected graph are always
connected by a tripod. This implies that minor models of subcubic graphs are topological
minor models.

The \emph{$k$-multiple} of a graph $G$ is the multigraph obtained by replacing each edge
$uv$ of $G$ by $k$ parallel edges between $u$ and $v$. To obtain a simple graph, it
suffices to subdivide each edge. We denote this graph by $G^{(k)}$.

A \emph{separation} is a pair $(A,B)$ of vertex sets, called \emph{sides},
such that $A \cup B = V(G)$, there are no edges between $A \setminus B$ and $B
\setminus A$. The \emph{order} of separation $(A,B)$ is $|A \cap B|$. The
separation is nontrivial if $A \subsetneq B$ and $B \subsetneq A$.
Given two vertex sets $X$ and $Y$, we denote by
$\mu(X,Y)$ the smallest order of a separation $(A,B)$ such that $X \subseteq A$
and $Y \subseteq B$. By Menger's theorem this is also the maximum number of
vertex-disjoint $X$-$Y$ paths. A graph is \emph{$k$-connected} if it has more
than $k$ vertices and no separation of order less than $k$.

Given a tree $T$ and an edge $uv \in E(T)$, we denote by $T^u_{uv}$ the subtree
$T'$ corresponding to the connected component of $T-uv$ containing $u$. Given three
nodes of $T$, their \emph{branching node} is the node which is at the intersection of the
paths between each pair. Assuming $T$ is rooted, the lowest common ancestor of two
nodes is their branching node with the root.

\subsection{SPQR trees}

An SPQR tree of a graph $G$~\cite{MacLanePlanar,SPQR1,SPQR2,SPQR3} is a tree $T$ such that
every node $u$ of $T$ is labelled with a multigraph $G_u$ and every edge $uv$ of $T$ is
mapped to an edge $e_u$ of $G_u$ and an edge $e_v$ of $G_v$, with the constraint that, for
every node $u$, an edge of $G_u$ is mapped to at most one edge of $T$ (these edges of
$G_u$ are called \emph{virtual edges}). A node $u$ of an SPQR tree can be of one of four
types, each defining different types of label multigraph $G_u$:
\begin{itemize}
\item if $u$ is an $\Series$-node, then  $G_u$ is a simple cycle of size at least three;
\item if $u$ is a $\Parallel$-node, then $G_u$ is a dipole (the graph on two vertices with at least three multiple edges between them);
\item if $u$ is a $\Q$-node, then  $G_u$ is the single edge graph;
\item if $u$ is an $\Rigid$-node, then $G_u$ is a simple $3$-connected graph.
\end{itemize}

The SPQR-tree $T$ of a $2$-connected graph $G=(V,E)$ is computed by: recursively finding a
separating pair $\{x,y\}$ of vertices, identifying two subgraphs $G_1=(V_1,E_1)$ and
$G_2=(V_2,E_2)$ of $G$ such that $V=V_1\cup V_2$, $V_1\cap V_2=\{x,y\}$ and $E$ is
partitioned into $E_1$ and $E_2$; then two adjacent nodes $u_1$ and $u_2$ are created and
respectively labelled by $G_1$ and $G_2$, each augmented with the virtual edge $xy$ (observe
that new virtual edges may generate multiple edges). We say that the virtual edge $xy$ is
mapped to the tree edge $u_1u_2$. The decomposition process continues on until the
resulting nodes are $\Series$-nodes, $\Parallel$-nodes or $\Rigid$-nodes, yielding a
canonical SPQR tree of $G$. 
Let us remark that if $G=(V,E)$ is a $2$-connected graph,
then $T$ does not contain any $\Q$-node.  Moreover, no pairs of $\Series$-node or of
$\Parallel$-nodes are adjacent in $T$. By construction, there is a bijection between the
virtual edges and the tree edges of $T$, and for every edge $xy\in E$, there exists a
unique node $u$ of $T$ such that $G_u$ contains the edge $xy$.
Let us state a few more useful observations on SPQR trees. 

\begin{observation} \label{obs_SPQR}
Let $T$ be the SPQR tree of a ($2$-connected) graph $G=(V,E)$. Then 
\begin{enumerate}
\item $(T,\beta)$, with $\beta(t)=V(G_t)$ for $t \in V(T)$, is a tree decomposition of $G$ such that every adhesion has
size two;
\item if $xy$ is a virtual edge of the graph $G_u$ associated to the edge $uv$
of $T$, then $G$ contains an $x,y$-path whose inner vertices are all contained
in the subtree of $T - u$ incident to $uv$ from the vertices of $G$ belonging
to $G_u$;
\item for every node $u$ of $T$, the multigraph $G_u$ is a topological minor of $G$.
\end{enumerate}
\end{observation}

\subsection{Tree decompositions}\label{subsec:treedec}

A \emph{tree decomposition} of a graph $G$ is a pair $(T,\beta)$, where $T$ is a tree and
$\beta: V(T) \to 2^{V(G)}$ is a mapping of subsets of $|V(G)|$ called \emph{bags} to each
node of $T$ satisfying the following conditions:
\begin{enumerate}[label=(T\arabic*)]
\item\label{enum:vertex-treedec} for every vertex $v \in V(G)$, the set $\{t \in V(T) | v \in \beta(t)\}$ induces a
nonempty (connected) subtree of $T$.
\item\label{enum:edge-treedec} for every edge $uv \in E(G)$, there is a node $t \in V(T)$ such that $\{u,v\}
\subseteq \beta(t)$.
\end{enumerate}

The \emph{width} of a tree decomposition is $\max_{t \in V(T)} |\beta(t)|-1$. The
treewidth of a graph $G$ is the minimum width over its tree decompositions, and we denote it
by $\tw(G)$.

It will often be practical to fix a mapping $t:E(G) \to V(T)$, such that for each edge
$uv \in E(G),$ we have $\{u,v\} \subseteq \beta(t(uv))$. Such a mapping always exists by
the defining property \ref{enum:edge-treedec}. In particular, it makes it simpler to take
into account the existence of edges of the graph when parsing the tree decomposition. We
then say that edge $uv$ is \emph{introduced} in (bag) $t(uv)$ or that $t(uv)$ contains the
introduction of $uv$.

We denote by $\alpha:E(T) \to 2^{V(T)}$ the \emph{adhesions} of the decompositions,
defined by $\alpha(st):= \beta(s) \cap \beta(t)$. We may abusively refer to the nodes of a
tree decomposition as bags and to its edges as adhesions. The notation $\beta$ is extended
to subtrees of $T$ by considering the union of the bags contained in the subtree.

The torso at node $t$ is the graph $G_t$ on vertex set $\beta(t)$ with edges corresponding
to $G[\beta(t)]$ and cliques covering adhesions $\alpha(tt')$ for each $tt' \in E(T)$.
Note that the multigraphs that label $\Rigid$-nodes and $\Series$-nodes of an SPQR tree
are torsos.

The following properties can always be enforced by a preprocessing. The first three are
relatively standard and can be enforced in linear time. The last one can be enforced by
rerooting subtrees of the decomposition if two equal adhesions appear non-consecutively in
the tree decomposition. This may require an additional polynomial factor in the maximum
adhesion size (which will usually be bounded).

\begin{enumerate}[label=(T\arabic*)]
\setcounter{enumi}{2}
\item\label{enum:irreducible-treedec} $\forall t,t' \in V(T), \beta(t) \neq \beta(t')$
\item\label{enum:conn-treedec} $\forall t \in V(T), \forall T' \in cc(T-t), G[\beta(T')
\setminus \beta(t)]$ is connected
\item\label{enum:adh-neighbour} $\forall tt' \in E(T), v \in \alpha(tt') \Rightarrow
\forall T' \in cc(T-tt'), N(v) \cap (\beta(T_i) \setminus \alpha(tt')) \neq \varnothing$
\item\label{enum:unique-adh} $\forall s,t,t' \in V(T), \alpha(st)=\alpha(st') \Rightarrow
\alpha(st)=\alpha(st')=\beta(s)$

\end{enumerate}

A tree decomposition is \emph{$k$-lean} if it satisfies:
\begin{enumerate}[label=(T\arabic*)]
\setcounter{enumi}{6}
\item\label{enum:k-lean-adh} $\forall st \in E(T), |\alpha(st)| \leq k$
\item\label{enum:k-lean-menger} for all $t_1,t_2 \in V(T)$, $X_1 \subseteq \beta(t_1)$,
$X_2 \subseteq \beta(t_2)$, if there exists a separation $(A,B)$ with $X_1 \subseteq A$
and $X_2 \subseteq B$ of order at most $\min(k,|X_1|,|X_2|)$, then there exists $e \in
E(T)$ on the path between $t_1$ and $t_2$ in $T$ with $\alpha(e)$ a minimal separator and
$|\alpha(e)| \leq p$.
\end{enumerate}

If a tree decomposition satisfies all of the above conditions, we call it a
\emph{well-formed $k$-lean} tree decomposition. Importantly, such a decomposition exists
for all graphs (unlike tree decompositions of bounded width).

\section{Using SPQR trees for the easy cases of the gem and planar graphs}
\label{sec:spqr}

The proof of \cref{thm:tbw-obs} revolves around understanding the
structure of graphs that do not contain large fan topological minors. We begin
with two introductory cases: the gem, which is the smallest fan which is not
subcubic, and the class of planar graphs. In
both cases, we use SPQR trees which, combined with block-trees, can be seen as
canonical tree decompositions with adhesions of size at most two. Specific
properties of the SPQR tree of a planar graph's $2$-connected components allow
us to approximate the size of a fan topological minor. Using these properties,
we can in turn construct tree-layouts of bounded bandwidth. This is meant as an
introduction to the characterisation of graphs excluding a fan topological minor with
simpler proofs.

\subsection{Forbidding a gem}\label{sec:forb-gem}

We first consider the case of excluding the fan $F_4$, also called the
\emph{gem} graph, for which we provide an exact characterisation using
SPQR trees.

\begin{theorem} \label{th_gem}
A simple graph $G$ excludes $F_4$ as a topological minor if and only if
every $2$-connected component $C$ of $G$ has an SPQR tree $T$ with the following
properties:
\begin{itemize}
\item[(1)] for every $\Rigid$-node $t$ of $T$, $\Delta(G_t)\leq 3$. 
\item[(2)] for every vertex $u$ of $C$, $T$ does not contains two nodes $s$ and $t$ that are $\Parallel$-nodes or $\Rigid$-nodes, and such that $u$ is a vertex of $G_s$ and of $G_t$.
\end{itemize}
\end{theorem}

The first condition of Theorem~\ref{th_gem} shows that excluding $F_4$ imposes
a bound on the maximum degree of the torsos of $\Rigid$-nodes. As we will see
this property will generalise for larger fans in the planar case. The
second property indicates that a given vertex cannot appear in too many
separating pairs. This also generalises to planar graphs. In the case of large fans in
non-planar graphs, we will have to use more complicated properties because the
shape of a fan model in a tree decomposition will be significantly less
constrained. We will, however, show that it is possible to compute tree decompositions
with bounded degrees in bags and bounded spans of vertices.

\begin{proof}
Since the gem is a $2$-connected graph, any subdivision of a gem is
contained in a $2$-connected component of $G$. So let us assume without
loss of generality that $G$ is $2$-connected and let $T$ be an SPQR tree of $G$.
\begin{description}[listparindent=1.5em]
\setlength{\itemsep}{7pt}
\item[($\Rightarrow$)] We first make the following useful claim. Note that the
existence of three neighbours is not an hypothesis since $3$-connectedness
implies degree at least $3$.
\begin{claim}\label{claim:ordered-f3}
For any vertex $u$ in the torso $G_t$ of an $\Rigid$-node $t$, and any three
neighbours $v_1,v_2,v_3$ of $u$ in $G_t$, there is a subdivided $F_3$ subgraph
in $G_t$ with universal vertex $u$, and whose subdivided $P_3$ visits
$v_1,v_2,v_3$ in this order.
\end{claim}
\begin{subproof}
The graph $G_t$ contains edges $uv_1$ and $uv_3$. Since $G_t$ is $3$-connected,
$G_t$ contains two internally vertex-disjoint paths $P_1,P_2$ from $v_2$ to
respectively $v_1$ and $v_2$, with both $P_1$ and $P_2$ avoiding $u$. Now
observe that the union of the edges $uv_1,uv_2,uv_3$ and the paths $P_1$ and
$P_2$ correspond to the required $F_3$.
\end{subproof}

We show property (1) by contraposition. Suppose that $t$ is an
$\Rigid$-node of $T$ containing a vertex $u$ of degree at least four. Let
$v_1,v_2,v_3,v_4$ be neighbours of $u$ in $G_t$. We apply
\cref{claim:ordered-f3} to $u$ and $v_1,v_2,v_3$ to obtain $H$, a subdivided
$F_3$ on $u,v_1,v_2,v_3$ in $G_t$.
Since $G_t$ is $3$-connected, there exist two vertex-disjoint paths
from $v_4$ to $H$, both avoiding $u$. Being vertex-disjoint, it must be that
one of these two paths, denote it by $P'$, does not end in $v_2$. We conclude that
the union of $H$ and $P'$ is a subdivided $F_4$ contained in $G_t$.
This implies that $G$ also contains a subdivided $F_4$ because $G_t$ is a
topological minor of $G$.

Using \cref{claim:ordered-f3}, we deduce that, for any $\Rigid$-node $t$, for
any pair of adjacent vertices $x,y$ of $G_t$, there is a subdivided $F_3$ in
$G_t$ such that $x$ is the center of the fan, and $y$ is the first vertex of
its path.
Similarly, for any $\Parallel$-node $t$, for any edge $e=xy$ of $G_t$, there is
a subdivided $F_3$ in $G$ (assuming $G$ is simple) whose model has the first
vertex of its path on the other side of the adhesion of $T$ corresponding to
$e$.
It remains to show that, if a vertex $u$ of $G$ is contained in the torsos of two
nodes of type $\Rigid$ or $\Parallel$, then we can combine the models of $F_3$
centered on $u$ obtained from each node.

Without loss of generality, we consider two such nodes $s,t$ that are either adjacent
in $T$ or with exactly one $\Series$-node $r$ between them in $T$. Let
$e_1=uy_1$ and $e_2=uy_2$ be the virtual edges corresponding to adhesions in
the $s$-$t$ path in $T$ which are incident to $s$ and $t$ respectively. In all
cases, we can obtain a $y_1$-$y_2$ path $P^\Series$ that is internally vertex
disjoint from vertices of $G_s$ and $G_t$. Let $Q_s,Q_t$ be the vertex sets of
models of $F_3$ centered on $x$, using $e_1$ and $e_2$ respectively to link to the
first vertex on their path, obtained from $s$ and $t$. 
Let $Q'_s$ and $Q'_t$ correspond to the subgraphs of $G$ given by the models
$Q_u$ and $Q_v$, without their part crossing the adhesions $\{u,y_1\}$ and
$\{u,y_2\}$ respectively. Observe that $Q'_s \cup Q'_t \cup P^\Series$ is a
subdivided $F_4$ centered on $u$, see \cref{fig:glue-F3}. In particular, we
cannot obtain a dipole from this construction because either one of $s$ and $t$
is an $\Rigid$-node, or $P^\Series$ contains an edge.

\item[($\Leftarrow$)] Assume that conditions (1) and (2) hold and suppose that $G$
contains a subgraph $H$ that is a subdivision of $F_4$. Let $u$ be the degree
$4$ vertex of $H$. Since no pair of $\Series$-nodes are adjacent, every vertex
not appearing in any $\Parallel$-node nor $\Rigid$-node has degree $2$ in $G$.
Hence, by (2), $T$ contains exactly one node $t$ that is a $\Parallel$-node or an
$\Rigid$-node such that $u$ is a vertex of $G_t$. By (1), $t$ is not an
$\Rigid$-node as it would contradict $u$ having degree at least $4$. Thereby,
$t$ is a $\Parallel$-node such that the vertex set of $G_t$ is $\{u,v\}$. It
follows that $v$ is the unique vertex of $G$ such that $G$ contains three
vertex-disjoint paths between $u$ and $v$. In turn, this implies a
contradiction since there exist two vertices with this property in the subgraph
$H$ of $G$.\qedhere
\end{description}
\end{proof}

\begin{figure}
\centering
\includegraphics{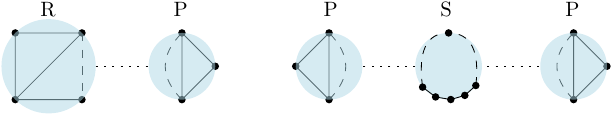}

\caption{The two main cases when constructing a gem that is not fully contained in an $\Rigid$-node.}
\label{fig:glue-F3}
\end{figure}

\subsection{Treebandwidth of planar graphs}

Let us now asymptotically characterise the presence of a large fan topological
minor in a planar graph. As in Theorem~\ref{th_gem}, we retrieve one condition
bounding the maximum degree of the torsos of $\Rigid$-nodes and another one
limiting the \emph{span} of a vertex over the separating pairs.

\begin{theorem}\label{thm:planar-structure}
Let $T$ be an SPQR tree of a $2$-connected planar graph $G=(V,E)$, and let $k \geq 3$. If $G$ does not contain $F_k$ as a topological minor, then:
\begin{itemize}
\item[(1)] for every $\Rigid$-node $t$, $\Delta(G_t)<k$;
\item[(2)] for every vertex $x$ of $G$, the number of separating pairs containing $x$ that are mapped to edges belonging to a common path of $T$ is at most $2k+1$.
\end{itemize}
Conversely, if conditions (1) and (2) hold, then $G$ does not contain $F_{2k^2}$ as a topological minor.
\end{theorem}

\begin{proof}

\begin{description}[listparindent=1.5em]
\setlength{\itemsep}{7pt}
\item[($\Rightarrow$)]
Let us first prove condition (1). For now, let us assume that $G$ is $3$-connected.

\begin{claim}\label{lem:3conn-planar-wheel}
Every vertex $x$ of a $3$-connected planar graph is the center of a subdivision of a $\deg(v)$-wheel.
\end{claim}

\begin{subproof}
Let $\Sigma$ be a plane embedding of $G$, and $v$ be an arbitrary vertex of
$G$. Let $H$ be the subgraph of $G$ induced by the set of vertices $V_H=\{x\in
V\mid \exists f\mbox{ a face of }\Sigma \mbox{ and } x,v\in\beta(f)\}$. We
prove that $H$ is a subdivision of a wheel. To that aim we consider a face $f$
incident to $v$, that is $\beta(f)=\langle x_1,x_2, \dots,
x_k\rangle$ with $k>3$ and $v=x_1$. We show that only $v$, $x_2$ and $x_k$
may appear in other faces incident to $v$. For the sake of contradiction,
suppose that there exists $x_i$ with $2<i<k$ such that $x_i$ is a vertex on the
boundary on some other face $f'\neq f$ incident to $v$. We consider the boundary $\beta(f')=\langle
x'_1,x'_2,\dots,x'_{k'}\rangle$ with $v=x'_1,x'_j=x_i$ (assume coherent cyclic
orderings of the faces). Since $f$ and $f'$ are faces of $\Sigma$, there is
no path from $x_2$ to $x_k$ that avoids $v$ and $x_i=x'_j$, see
Figure~\ref{fig:wheel}. Indeed, $x_2$ is embedded inside the cycle 
$C=\langle v=x_1,\dots,x_i=x'_j,\dots,x'_{k'} \rangle$ and the only edges
incident to $C$ that are embedded outside of $C$ in $\Sigma$ are incident to
$v$ and $x_i=x'_j$. In other words, the pair $\{v,x_i\}$ is a separator for the
vertices $x_2$ and $x_k$, which contradicts $G$ being $3$-connected.
\end{subproof}

\begin{figure}[h]

\centering
\includegraphics{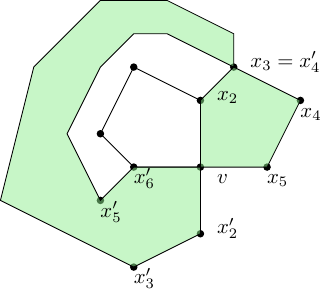}

\caption{Illustration of the contradiction if two faces share two vertices that are not consecutive.}
\label{fig:wheel}
\end{figure}

Since a $k$-wheel contains $F_k$ as a subgraph, we conclude that if $G$ is a
$3$-connected planar graph excluding $F_k$ as a topological minor, then
$\Delta(G)<k$, implying condition (1).

Let us now prove condition (2).
Let $P$ be a path of $T$ containing at least $2k$ nodes such that a vertex $x$ of $G$ belongs to every separating pair corresponding to an edge of $P$.
Let $r$ be a node of $T$, and let $s$ and $t$ be its two neighbours on $P$. We let $V^r$ denote the vertices of $G$ appearing in the nodes of $T_{rs}^s$ or in the nodes of $T_{rt}^t$. We extract a fan from $G$ by considering the different types of nodes on $P$.
\begin{itemize}
\item Suppose that $r$ is a $\Parallel$-node. From \cref{obs_SPQR}, 
if $x,y_i$ are the two vertices of $G_r$, then $G$ contains at least three internally vertex-disjoint $x$-$y_i$ paths and one of them, call
it $P_i$, is internally vertex-disjoint from the vertices of $V^r$.

\item Suppose that $r$ is an $\Series$-node and let $\{x,y_i\}$ and $\{x,y_{i+1}\}$ be the two separating pairs corresponding to
the edges $rs$ and $rt$ of $P$. Then, from \cref{obs_SPQR}, $G$ contains a
$y_i$-$y_{i+1}$ path $Q_i$ that is internally vertex-disjoint from the vertices
of $V^r$.

\item Suppose that $r$ is an $\Rigid$-node and let $\{x,y_i\}$ and $\{x,y_{i+2}\}$ be the two separating pairs corresponding to
the edges $rs$ and $rt$ of $P$. Since $x$ has degree at least $3$ in $G_r$, it has a neighbour $y_{i+1}$ distinct from $y_i$ and $y_{i+2}$.
Since $G_r$ is $3$-connected, using \cref{obs_SPQR} to lift paths from $G_r$ to $G$, $G$ contains three internally vertex-disjoint paths $P_{i+1}$, $Q_i$ and $Q_{i+1}$ 
respectively from $x$ to $y_{i+1}$, $y_i$ to $y_{i+1}$ and from $y_{i+1}$ to $y_{i+2}$, with $P_{i+1},Q_i,Q_{i+1}$ internally vertex-disjoint from the vertices of $V^r$.
\end{itemize}

Since there cannot be two consecutive $\Series$-nodes on $P$ and $P$ contains
at least $2k+1$ nodes, it contains at least $k$ nodes that are $\Parallel$ or
$\Rigid$. From the arguments above, we have identified
at least $k$ vertices $y_1,\dots, y_k$ of $G$ and pairwise internally vertex-disjoint
paths $Q_i$ from $y_i$ to $y_{i+1}$ and $P_i$ from $x$ to $y_i$.
Together, these paths witness the existence in $G$ of a
subdivided $k$-fan\footnote{We can obtain a $(k+2)$-fan by adding an
$x$-$y_1$ path and an $x$-$y_k$ path.}, implying that (2) holds.

\item[($\Leftarrow$)] Conversely, suppose that conditions (1) and (2) hold. 

\begin{claim}
The vertices of degree at least $3$ in an $F_{\ell}$ topological minor in $G$ are all introduced along a path of the SPQR tree.
\end{claim}

\begin{subproof}
Let $x$ denote the vertex of degree $\ell$ of an $F_{\ell}$ model in $G$, and
$Y$ the set of its vertices of degree $3$.
Assume towards a contradiction that the property does not hold, then there is a
node $t$ of the SPQR tree $T$ such that there are at least $3$ subtrees of
$T - t$ that contain (distinct) vertices of $Y \cup \{x\}$ that are not contained in $G_t$.
Each subtree contains at least one vertex of $Y$, as otherwise, there would be a
separator of order $2$ between $x$ and $Y$. One of the subtrees $T'$ contains a
vertex $y$ of $Y$ that appears in the fan model's path between two vertices of
$Y$ not contained in $T'$. We deduce that there are two vertex-disjoint paths
that leave the subtree.
Since they cover the adhesion of size $2$ incident to $t$, and $y$ also has a
vertex-disjoint path to $x$, $x$ must appear in $T'$ and not $t$. However, this
implies that the vertex-disjoint paths from $x$ to the vertices of $Y$ outside
$T'$ should be hit by the same adhesion between $t$ and $T'$, a contradiction.
\end{subproof}

We can now directly use the conditions (1) and (2). Our fan model contains at most
$(2k+1)(k-1)$ vertices of degree $3$.
This means we cannot find a model of $F_\ell$ for $\ell=(2k+1)(k-1)+3 \leq 2k^2$.\qedhere
\end{description}
\end{proof}

Note that, in the above proof, planarity is only used to bound the degree in $\Rigid$-nodes,
and the rest of the arguments do not require planarity. However, this degree bound is
crucial to use the following immediate consequence of the construction of
\cite{tpwDegree,tpwWood}. This construction can be efficiently implemented
\cite{algo-tpw}, see \cref{thm:planar-algo} for details.

\begin{lemma}\label{tpw-root}
Given a graph $G$, there is a tree-partition of width $O(\Delta(G)\tw(G))$.
Moreover, we can enforce any subset of $O(\Delta(G)\tw(G))$ vertices to be in
the root bag of the tree-partition.
\end{lemma}

By combining the structure of SPQR trees from Theorem~\ref{thm:planar-structure}
with known constructions for tree-partitions of width $O(\Delta(G) \tw(G))$
(see Lemma~\ref{tpw-root}), we manage to prove that the wall number and
the fan number characterise treebandwidth on planar graphs.
More precisely, we obtain the following bound.

\begin{lemma}\label{lem:planar-construction}
Let $G=(V,E)$ be a planar graph. If $\tw(G)\leq w$ and $G$ does not contains a $k$-fan as
topological minor, then $\tbw(G) = O(k\cdot (k+w))$.
\end{lemma}

\begin{proof}
Let $G$ be a planar graph with no $F_k$ topological minor and such that $\tw(G) \leq w$.

If $G$ is $3$-connected, it has degree less than $k$ (see Theorem~\ref{thm:planar-structure}, condition (1)). 
As such, it admits a tree-partition of width $O(kw)$ from which we may deduce a
tree-layout of bandwidth $O(kw)$ by rooting arbitrarily and replacing each bag by a linear
ordering.

If $G$ is $2$-connected, it admits an SPQR tree $T$. We produce the following layouts for each
type of node:
\begin{itemize}
	\item For an $\Series$-node $t$, we produce a linear layout of bandwidth $2$ of $G_t$ (a
	cycle).
	\item For an $\Rigid$-node $t$, we produce a tree-layout of $G_t$ of bandwidth $O(kw)$ using the
	construction mentioned above. 
	\item For a $\Parallel$-node $t$, we produce any linear layout of its $2$ vertices.
\end{itemize}

We will now glue such layouts along the separating pairs of the SPQR tree. Observe that
separating pairs are adjacent in the torsos of nodes of the SPQR tree. In particular,
their distance in layouts of torsos is bounded. We can always enforce two adjacent
vertices $x,y$ to be at the root of the layout of an $\Series$-node. Finally, in the
$\Rigid$-node case, we can ask for two vertices $x,y$ to be at the root of the
tree-layout, and for the vertices of $N[x] \cup N[y]$ to appear in the topmost $2k$ vertices of
the tree-layout using \cref{tpw-root}.

We proceed as follows. We root $T$, and produce a tree-layout for its root node
as described above. Then, we recursively compute tree-layouts for each subtree of $T$ by
computing a tree-layout of their root node. We ensure that the two
vertices of the adhesion incident to a subtree are placed at the root of the
tree-layout of the subtree.

To glue tree-layouts of adjacent nodes of $T$, we attach the layout of a
child node to its parent node as follows.
Let $x,y$ be the two vertices in common. Without loss of generality, assume that
$y$ is deepest in the layout of the parent. Compute a layout of the child with
$x$ at the root and $y$ as its child.
Remove $x$ from the layout of the child node and then glue the layouts at $y$.

The bandwidth of the produced layout can be bounded by combining condition (2) of Theorem~\ref{thm:planar-structure} and the
properties of tree-layouts computed for each node of $T$. We distinguish two cases. First,
two neighbours $a,b \in V(G)$ may be introduced in the same node of $T$ and their distance in
the tree-layout is at most $O(kw)$. Second, they can be introduced in distinct nodes
$t_a,t_b$ and then the vertex $a$ (introduced first without loss of generality) is
contained in all adhesions between $t_a$ and $t_b$. There are at most $2k$ such adhesions
due to condition (2). Moreover, since each adhesion corresponds to an edge incident to $a$ in
the torsos, the tree-layouts of child nodes along this sequence of adhesions are always
glued within the first $2k$ vertices. This sums up to a distance of $4k^2$ for all nodes
except $t_a$. The distance between the two vertices in the adhesion incident to $t_a$ is
bounded by $O(kw)$. We conclude that the bandwidth is at most $O(k(k+w))$.

If the graph $G$ is not $2$-connected, we can compute the block-cut tree, root
it arbitrarily, and compute a layout recursively starting from the root, by applying the
construction for $2$-connected graphs above to each block. The layout for the root block
can be rooted arbitrarily, and the layouts for other blocks are rooted at the cutvertex
they have in common with their parent block. In particular, this requires us to root the
SPQR tree of the block to a node containing said cutvertex, and to root the layout of this
node on the cutvertex. We then glue the layouts at their common vertex and the width is
just the maximum over the width of computed layouts for the blocks. Indeed, each edge is
contained in a single block.

Finally, if the graph $G$ is not connected, we can link the layouts arbitrarily
and the width is also the maximum width over the connected components.
\end{proof}

This translates into an approximation algorithm for treebandwidth when restricted to the class of planar graphs.

\begin{theorem}\label{thm:planar-algo}
Given a planar graph $G=(V,E)$ with $\tbw(G)\leq k$, a tree-layout of bandwidth
$O(2^{2k})$ can be computed in time $2^{O(k)}n$.
\end{theorem}
\begin{proof}
Fix a positive integer $k$, the parameter of our instance.

The block-cut tree and the SPQR tree of each block can be computed in linear time for any
graph \cite{SPQR1,SPQR3}. We then check if there is an $F_{2^k}$ topological minor by checking
that conditions (1) and (2) hold for $k'=2^k$. This can be carried out in linear time. If 
such a topological minor is found, then we may conclude that $\tbw(G)>k$.

It is straightforward to compute a layout for $\Series$-nodes and to glue tree-layouts.
However, to compute a tree-layout for $\Rigid$-nodes we invoke involved techniques to
compute tree-partitions of width $O(\Delta(H) \tw(H))$ in time $2^{O(k)}|V(H)|$ from \cite[Lemma
3.10]{algo-tpw}\footnote{The lemma was stated with an assumption of a polynomial bound on
$\Delta$ but it also holds with $\Delta=2^{O(k)}$ because the factor in the
running time which is a polynomial on $\Delta$ is now bounded by $2^{O(k)}$
and can be factored with the $2^{O(k)}$ factor (depending on the treewidth
bound), note however that this affects the base of the exponential in the
parameter in the running time.} (using data structures from
\cite{BodlaenderDDFLP16}). More precisely, we either construct tree-partitions
of width $O(k2^k)$ or conclude that $k < \tw(H) \leq \tw(G) \leq \tbw(G)$.

Finally, if we did not conclude that $\tbw(G) > k$, we constructed a tree-layout of
bandwidth $O(2^{k}(k+2^{k}))$ in time $2^{O(k)}|V(G)|$.
\end{proof}

We only superficially use planarity to obtain degree bounds. On the other hand,
our hardness proof and the NP-hardness proof of \cite{spanheight} produce
highly non-planar graphs. We may ask the following.

\begin{opquestion}
Is there an FPT (or even polynomial) algorithm for \textsc{Treebandwidth} on
the class of planar graphs?
\end{opquestion}


\section{Structure theorems}
\label{sec:structure}

We now move to the proof of the structure theorems for graphs excluding a
$k$-fan or a $k$-dipole as a topological minor. These theorems are based on
structural observations on well-formed $k$-lean tree decompositions. In particular, we
enforce property \ref{enum:unique-adh} by introducing some branching bags between
adhesions on the same vertex set. This corresponds to a weakened generalisation of
$\Parallel$-nodes of an SPQR tree (and of cutvertices). For both
theorems, we separate these structural observations into lemmata that are
abstracted away from their application in the proof of the main theorems.

\subsection{Excluding a fan}

We start by characterising the implications of forbidding large subdivided fans centered
on a fixed vertex $v$. This is done by leveraging recent work on rooted minors. We can reduce our considerations on fans centered on $v$ to rooted path
minors.

\begin{observation}
	There is no subdivided $k$-fan centered on vertex $v$ if and only if
there is no $N(v)$-rooted $P_k$ minor in $G-v$.
\end{observation}

The size of a $U$-rooted $P_k$ minor can be approximated via the treedepth of
$U$ \cite{quickapexforest}, which is defined recursively as follows:

\begin{itemize}
\item $\td(G,U)$ is $0$ if $V(G)\cap U = \varnothing$.
\item If $G$ is not connected, then $\td(G,U) = \max_{C \in cc(G)} \td(C,U)$.
\item If $G$ is connected, then $\td(G,U)= 1+\min_{x \in V(G)} \td(G-x,U)$.
\end{itemize}

A forest $F$ on a subset of $V(G)$ is said to be an \textit{elimination forest of $U$ in
$G$} if $F$ contains all the vertices of $U$, all the edges of $G[F]$ have an
ancestor-descendant relationship in $F$, and, for every connected component $C$ of $G
\setminus F$, $N(C)$ is contained in a single root-leaf path of $F$. As in the usual
treedepth, the treedepth of $(G,U)$ can also be defined as the minimum height of an
elimination forest of $U$ in $G$. Moreover, if $G$ is connected, then this minimum is obtained on a tree. 

The following bound is proved in \cite{quickapexforest}.
\begin{lemma}\label{lem:fan-to-td}
	If there is no $U$-rooted $P_k$ minor in $G$, then $\td(G,U) \leq \binom{k}{2}$.
\end{lemma}

We prove that the following weak converse holds.
\begin{lemma}\label{lem:td-to-fan}
	If $\td(G,U) \leq t$, then there is no $U$-rooted $P_{2^{t}}$ minor in $G$.
\end{lemma}

\begin{proof}
We prove the result by induction on the value of $\td(G,U)$. If $\td(G,U)=0$, then there
is no $U$-rooted $P_1$ minor in $G$ as $V(G)\cap U = \varnothing$.

Suppose now that $\td(G,U)=t+1$ and that the result holds for any $t' \leq t$. Without
loss of generality, we may assume that $G$ is connected since any $P_\ell$ minor model is
confined to a connected component. Then by definition of $\td(G,U)$, there exists a vertex
$v$ such that $\td(G,U)=1+\td(G-v,U)$. The $U$-rooted path minors that do not use $v$ for
their model, are also $U$-rooted minors of $G-v$ so we can upper bound their size by
$2^{t}-1$. If a $U$-rooted path minor of $G$ uses $v$, it is used in only one branch set
$X$. Hence, after removing the branch set $X$, we obtain two $U$-rooted path minors of
$G-v$. Their size is upper bounded by $2^t-1$, so our path minor of $G$ has size at most
$1+2(2^t-1)=2^{t+1}-1$.
\end{proof}

This allows us to conclude that the fan number is asymptotically equivalent
to a parameter we call \emph{neighbourhood treedepth}.

\begin{theorem}\label{thm:neighbourhood-td}
	The size of the largest fan topological minor in $G$ is tied to 
	$$\max_{v \in V(G)} \td(G,N[v]).$$
\end{theorem}

We now move our attention to $k$-lean tree decompositions. Our structure theorem is a
consequence of the two following key lemmata about the treedepth of sets and tree
decompositions. 

\begin{lemma}\label{lem:bound-to-td}
Given a tree decomposition $(T,\beta)$ of $G$ with adhesions of order at most
$k$ and a set $U$ of vertices such that, for each bag $\beta(t)$ of size more
than $k$, 
\begin{itemize}
\item the number of subtrees $T'$ of $T-t$ with
$(\beta(T')\setminus\beta(t))\cap U \neq \varnothing$ is bounded by $a$,
\item $|\beta(t) \cap U|$ is bounded by $b$, 
\end{itemize}
and for each path $P$ of $T$, there is a subset $\mathcal{B}$ of at most $c$ bags such that
$\beta(P) \cap U \subseteq \beta(\mathcal{B})$, and such that bags incident to subtrees
$T'$ of $T-P$ such that $(\beta(T')\setminus\beta(P))\cap U \neq \varnothing$ are in
$\mathcal{B}$.

Then, $\td(G,U) \leq c(b+ak)$
\end{lemma}

\begin{proof}

Let $\overline{\mathcal{B}}$ be a minimum set of bags closed by branching node (i.e. the
bag of the branching node of each triple of bags in the set is also in the set) that
covers $U$. In particular, by minimality, the number of bags of $\overline{\mathcal{B}}$
on any path $P$ of $T$ is bounded by $c$. Let $T_{\overline{\mathcal{B}}}$ be the tree on
the set of nodes whose bags are in $\overline{\mathcal{B}}$ obtained from $T$ by removing
subtrees of $T-V(T_{\overline{\mathcal{B}}})$ that are incident to only one node
of $T_{\overline{\mathcal{B}}}$, and replacing subtrees of
$T-V(T_{\overline{\mathcal{B}}})$ incident to two nodes $s,t$ by an edge $st \in
E(T_{\overline{\mathcal{B}}})$. We construct an elimination tree for $U$ from
$T_{\overline{\mathcal{B}}}$. Pick an arbitrary root and then replace each bag $\beta(t)$ of
$T_{\overline{\mathcal{B}}}$ by a linear ordering which begins with the following
set $\beta_U(t)$ of vertices: vertices of $U \cap \beta(t)$ and vertices of $\alpha(tt')$ for each $tt' \in E(T)$ such that
the subtree $T'=T^{t'}_{tt'}$ satisfies $(\beta(T')\setminus\beta(t))\cap U \neq
\varnothing$.

Our assumptions allow us to bound the number of vertices in $\beta_U(t)$ by $b+ak$.
We remove from the linear ordering of $\beta(t)$ any vertex that appears in the bag of its
parent in $T_{\overline{\mathcal{B}}}$.
We obtain an elimination tree from these linear orderings by attaching the linear ordering
of $\beta(t)$ to the last vertex of $\beta_U(s)$ where $s$ is its parent. The depth of
vertices of $U$ in this elimination tree is bounded by $c(b+ak)$. It is a valid
elimination tree since connected components of vertices not placed in some $\beta_U(t)$ are
either adjacent to a subset of some $\beta_U(t)$ or to a subset of $\beta_U(s) \cup
\beta_U(t)$ for some $st \in E(T_{\overline{\mathcal{B}}})$. In all cases such vertices
are on a common root-to-leaf path of the constructed elimination tree.
\end{proof}

\begin{lemma}\label{lem:td-to-bound}
Given a well-formed $k$-lean tree decomposition $(T,\beta)$ of $G$ and a set $U$ of vertices.
If $\td(G,U) \leq k$, then
\begin{itemize}
\item for all $t \in V(T)$, either $\beta(t)$ has size at most $k$, or $|\beta(t) \cap U|$
and the number of subtrees $T'$ of $T-t$ such that $(\beta(T')\setminus\beta(t))\cap U
\neq \varnothing$ are bounded by a function of $k$; and
\item for all paths $P$ of $T$, $|\beta(P) \cap U|$ and the number of bags $\beta(t)$ of
$P$ incident to subtrees $T'$ of $T-P$ such that $(\beta(T')\setminus\beta(t))\cap U \neq
\varnothing,$ are bounded by a function of $k$.
\end{itemize}
\end{lemma}

\begin{proof}
	To show that the first property of the lemma is satisfied, assume that $t \in V(T)$ is
	such that $\beta(t)$ has size more than $k$. Since $\td(G,U) \leq k$, it means that
	there exists an elimination forest of depth at most $k$ for $U$ in $G$. In particular,
	any pair of vertices of $U$ not in ancestor-descendant relationship can be separated
	in $G$ with at most $k-1$ vertices by simply taking the path from one of their lowest
	common ancestor to the root of its tree in the forest. Now if $|\beta(t) \cap U| \geq
	k+1$, then we can separate two of these vertices with at most $k-1$ vertices, which
	contradicts the definition of $k$-lean. 

	Let us now show by contradiction that the number of subtrees $T'$ of $T-t$ such that
	$(\beta(T')\setminus\beta(t))\cap U \neq \varnothing$ is bounded by $2^{k+1}$.
	Since all the adhesions are distinct (by property \ref{enum:unique-adh} and
	$|\beta(t)|>k$), it means we can find $k+1$ vertex-disjoint $\beta(t)$-$U$ paths
	that are internally disjoint with $\beta(t)$. Let $X$ be the set of endpoints
	of these paths in $\beta(t)$, and $U'$ the set of endpoints of these paths in $U$. We
	will now use the elimination tree to find a separator of $U'$. Consider a
	separator $S$ obtained by iteratively extending a root path in the elimination
	tree to the child such that the subtree below contains at least one of the elements
	of $U'$ whose path to $X$ is not hit by $S$. This subtree is unique since there cannot
	be two elements of $X$ separated by $S$ with $|S| \leq k-1$ due to $(T,\beta)$ being
	$k$-lean. Thus $S$ reaches a leaf, and we constructed a set $S$ of size $k$ that hits
	all $X$-$U'$ paths, but these paths are vertex-disjoint and there are $k+1$ of them.
	
	Let us now show that the second property of the lemma is satisfied. To do so, let us fix a path $P$ of $T$. For convenience, we consider $P$ to be rooted at one of its endpoints, and, if a
bag $\beta(t)$ is the closest bag to the root containing some vertex $u$, that
$u$ is \emph{introduced} in $\beta(t)$.
First, we note that we can obtain a path that is internally vertex-disjoint
with vertices of bags of $P$ for each adhesion incident to $P$ that is branching to
neighbours of $v$ not in $\beta(P)$ using \ref{enum:conn-treedec}. Furthermore,
we may choose any vertex of the adhesion as the endpoint in $\beta(P)$ of such
a path. Let $A$ be the set of endpoint vertices of such paths in $\beta(P)$. If
there are at least $2^{f(k)}$ adhesions with distinct sets, each leading to neighbours of $v$,
we can ensure $|A| \geq f(k)$ for any function $f$. Let $X := A \cup (\beta(P) \cap U)$, note
that, since we know the first property of the lemma is satisfied, if $|X|$ is not bounded by a function of $k$, then the number of bags of
$P$ introducing vertices of $X$ is not bounded by a function of $k$. Remember also that by
property \ref{enum:unique-adh}, it almost holds that each adhesion in $P$ has a distinct vertex set. 

For every $s \leq k$, a \textit{section} of width $s$ is a subpath $P'$ of $P$ such that: 
\begin{itemize}
	\item Either the first element of $P'$ corresponds to the first element of $P$; or the first adhesion of $P'$ has size $s$; and
	\item either the last element of $P'$ corresponds to the last element of $P$; or the last adhesion of $P'$ has size $s$; and 
	\item Every other adhesion between elements of $P'$ has size greater than $s$
\end{itemize}
Consider $\mathcal{A}^s$ the set of adhesions of size $s$ in $P$, the sections of width $s$ are precisely the set of subpaths contained between consecutive adhesions of size $s$, or the extremities of the paths. 
We will show by induction on $s \leq k$, that there are at most $f(s,k)$ sections of $P$
of width $s$ that introduce vertices of $X$ for some function $f$. Let $q=2^k$, since
$\td(G,U) \leq k$, we know by \cref{lem:td-to-fan} that there is no $X$-rooted $P_q$ minor
in $G$.

Let us first show that $f(1,k) \leq q$. Indeed, the adhesions are distinct and it is
straightforward to connect an introduced vertex $x$ of $X$ to the two vertices $a,b$ in
adhesions of $\mathcal{A}^1$ that surround its introduction by an
$(x,a,b)$-tripod. The tripods are internally vertex-disjoint and their union is
easily seen to be an $X$-rooted path minor model which we know cannot be on $q$
vertices.

To understand the sections of width $s +1$, let us observe that the elements of
$\mathcal{A}^{s}$ induce a partition of the sections of width $s+1$ into blocks, which
corresponds to union of sections of width $s+1$ between two consecutive elements of
$\mathcal{A}^{s}$. Since the union of sections of width $s+1$ inside one block corresponds
to a section of width $s$, we know that only $f(s,k)$ different blocks introduce a vertex.
Therefore, if we can show that each block has a bounded number of sections of width $s+1$
introducing a vertex of $X$, we would be done. To do so, let us fix a block $\mathcal{B}$.
By Menger's theorem and the property of $k$-lean decomposition, we can find a collection
$Q_1, \dots, Q_{s+1}$ of $s+1$ paths linking the first adhesion of the block $\mathcal{B}$
to the last one. In particular, for each section $S$ of $\mathcal{B}$ this induces a set
of paths linking the vertices of the first and last adhesion. Some of the vertices of
these adhesions might be the same, in which case the path can be a single vertex, however
since adhesions are distinct\footnote{If they are not, the section is a bag of $T$ equal
to its two adhesions on $P$ and we may find a vertex-disjoint path from $U$ to any vertex
of the bag by \ref{enum:unique-adh} and \ref{enum:conn-treedec}. Furthermore, we are
guaranteed that the next adhesion on $P$ will be distinct.}, 
some of these paths must be non trivial. For a section $S$, if we denote by $\alpha(S)$
the vertices that belong to the first and last adhesion of $S$, we deduce from
\ref{enum:conn-treedec} that any vertex of $\beta(S) \setminus \alpha(S)$ can be linked to
some vertex of the first or last adhesion by a path disjoint from the $\alpha(S)$. In
particular, if a section introduces an element $x$ of $X$, then it can be attached to one
of the non trivial subpaths of the $Q_1, \dots, Q_{s+1}$ by a path $Q(x)$ internally
disjoint from the first and last adhesion of $S$. For any section $S$ of $\mathcal{B}$
that introduces an element of $X$, let us fix $x(S)$ any such element, and define $Q(S)$
as the path $Q(x(S))$ linking $x(S)$ to one of the non trivial subpaths of $Q_1, \dots,
Q_{s+1}$ inside $S$. Remark that if the extremity $y$ of $Q(S)$ belongs to the first
(resp. last) adhesion of $S$, then this is the first (resp. last) section that contains
$y$ in its adhesion. Therefore, for any element $y$ belonging to one of the paths $Q_1,
\dots, Q_{s+1}$, there exists at most two sections $S$ such that $y$ is the extremity of
$Q(S)$. From this, we deduce that if there exists more than $2q(s+1)$ sections of $S$ that
introduces an element of $X$, then $2q$ of the paths $Q(S)$ must have an endpoint on the
same path, and we can extract an $X$-rooted minor of length $q$, which is a contradiction.
Finally this proves that $f(s+1, k) \leq f(s,k) \cdot 2q(s+1)$.

We can conclude that $f(k,k) = 2^{O(k^2)}$ by combining the previous inequalities.
Now sections that are not separated by adhesions of size $t$ are bags of $P$
for which we can simply apply the first property of the lemma.
\end{proof}

Since the above lemma holds simultaneously for several sets $U$ and a fixed decomposition
$T$, in particular by considering the family of sets $\mathcal{U}=\{N[v]:v\in V(G)\}$, we
deduce a parameter equivalent to the largest fan topological minor based on observing
simple properties of well-formed $k$-lean tree decompositions. In the statement
below, we replace precise bounds by fixed functions $f_1,f_2$ independent of
the considered graph.

\begin{theorem}\label{thm:dec-fan}
    If there is no $F_k$ topological minor in $G$, then there exists
	$(a,b,c)=f_1(k)$ such that, for any $a' \geq a,$ in any well-formed $a'$-lean tree
	decomposition $(T,\beta)$ of $G$: 
	\begin{enumerate}
	\item each vertex $v$ has at most $b$ neighbours in the torso of each bag, and
	\item for every path $P$ of $T$, there are at most $c$ bags on $P$ that either contain
	the introduction of an edge incident to $v$, or are incident to an edge of $T-P$
	that leads to a bag containing the introduction to an edge incident to $v$.
	\end{enumerate}
	
	Conversely, if there exists $(a,b,c)$ satisfying the properties above, then
	there is no $F_{k'}$ topological minor in $G$ for some $k'=f_2(a,b,c)$. 
\end{theorem}

The statement of the theorem is slightly simplified compared to the earlier lemmata thanks to
the following consequence of properties \ref{enum:vertex-treedec},\ref{enum:edge-treedec},
and \ref{enum:adh-neighbour}.

\begin{observation}\label{obs:bound-to-torso}
In a well-formed $a'$-lean tree decomposition $T$, the set of neighbours of $v$ in the torso $G_t$
at $t \in V(T)$ is the union of $N_G(v) \cap \beta(t)$ and the adhesions $\alpha(tt')$
containing $v$. These adhesions are exactly the adhesions leading to subtrees $T'$
introducing a neighbour $w$ of $v$ in $G$ (i.e. $w \in \beta(T') \setminus \beta(t)$).
Such subtrees must contain the introduction of the edge $vw$.
\end{observation}

\begin{proof}[Proof of \cref{thm:dec-fan}]
\begin{description}[listparindent=1.5em]
\setlength{\itemsep}{7pt}
\item[($\Rightarrow$)] By \cref{thm:neighbourhood-td}, there is a $k'$ function of $k$
such that for all $v\in V(G),$ we have $\td(G,N[v]) \leq k'.$ Let $(T,\beta)$ be a well-formed
$a'$-lean tree decomposition, for $a' \geq a= k'$. We apply \cref{lem:td-to-bound} to
$(T,\beta)$ and the sets $N[v]$ for each $v \in V(G)$, which imply the desired properties.
We explicit how to obtain the properties claimed in the theorem's statement. If a bag $\beta(t)$ is
of size at most $k'$, then the torso $G_t$ has bounded size and its vertices have bounded
degree. Otherwise, we apply \cref{obs:bound-to-torso} to translate the bounds from
\cref{lem:td-to-bound}. We now fix a path $P$ of $T$. Since $|\beta(P)\cap N[v]|$ is
bounded by a function of $k$, in particular, there is a bounded number of introductions of
edges of $G$ incident to $v$ on the path $P$. There is also a bounded number of bags
$\beta(t)$ incident to subtrees $T'$ of $T-P$ that contain the introduction of edges $vw$
of $G$ incident to $v$ because such
subtrees must satisfy $(\beta(T')\setminus \beta(t)) \cap N(v) \neq \varnothing$ by property
\ref{enum:edge-treedec} applied to $vw$.

\item[($\Leftarrow$)] Consider a well-formed $a$-lean tree decomposition $(T,\beta)$ with
the given properties, and apply \cref{lem:bound-to-td} (using \cref{obs:bound-to-torso} to
translate the torso bound). We deduce that $\max_{v \in V(G)} \td(G,N[v])$ is bounded by a
function of $a,b,c$. We conclude by using \cref{thm:neighbourhood-td}.
\qedhere
\end{description}
\end{proof}

\subsection{Excluding a dipole}

Once again, we abstract the structure we are looking for by considering
how sets of vertices are arranged in a well-formed $k$-lean decomposition.

\begin{observation}\label{obs:dipole-to-sep}
	Given two vertices $u,v \in V(G)$, the size of the largest dipole centered
	on $u,v$ is tied to $\mu(N(u),N(v))$.
\end{observation}

Our structure theorem is the consequence of the following key lemmata. We
define the weight $w(t,U)$ of a set $U$ on a bag $\beta(t)$ to be the sum of
$|\beta(t) \cap U|$ and the minimum number of adhesions incident to $t$ that
hit all $\beta(t)$-$(U\setminus\beta(t))$ paths.

\begin{lemma}\label{lem:bound-to-sep}
Given a tree decomposition $(T,\beta)$ of $G$ with adhesions of order at most
$k$ and two sets $U_1,U_2$ of vertices such that for all $t \in V(T)$, either
$w(t,U_1) \leq a$ or $w(t,U_2) \leq a$, and there exists a set of $c$ bags
$\mathcal{B}$ such that $\beta(\mathcal{B})$ hits all $U_1$-$U_2$ paths.
Then, $\mu(U_1,U_2) \leq cak$.
\end{lemma}

\begin{proof}
For each bag in $\beta(t) \in \mathcal{B}$, let $A_i^t$ denote the union of
$\beta(t) \cap U_i$ and $\alpha(tt')$, for each $tt'$ in a minimum set $S$ of
adhesions such that $\alpha(S)$ hits all $\beta(t)$-$(U_i\setminus \beta(t))$
paths. By assumption, one of $A_1^t$ and $A_2^t$ is of size at most $ak$. Let
$A^t$ be the smallest of $A_1^t$ and $A_2^t$.

We now argue that the set $A$ which is the union of $A^t$ over bags $\beta(t)
\in \mathcal{B}$, is a $U_1$-$U_2$ separator. Assume towards a contradiction
that there is a $U_1$-$U_2$ path $P$ in $G-A$. Then, by our assumptions, $P$ is
hit by some $\beta(t) \in \mathcal{B}$. Moreover, it must be hit by $A^t
\subseteq A$ since $A^t$ will, either contain an endpoint of $P$, or an
adhesion $\alpha(tt')$ leading to a bag containing an endpoint of $P$, a
contradiction. We conclude since $|A| \leq cak$.
\end{proof}

\begin{lemma}\label{lem:sep-to-bound}
Let $(T,\beta)$ be a well-formed $k$-lean tree decomposition $(T,\beta)$ of $G$, and $U_1,U_2$ be two
sets of vertices such that $\mu(U_1,U_2) \leq c \leq k$. There exists
a set of $c$ bags $\mathcal{B}$ such that $\beta(\mathcal{B})$ hits all
$U_1$-$U_2$ paths, and for every $t \in V(T)$, either $w(t,U_1) \leq 2c$ or $w(t,U_2) \leq 2c$.
\end{lemma}

\begin{proof}
Let $S$ be a minimum $U_1$-$U_2$ separator, by our assumptions $|S| \leq c$.
For each vertex $x \in S$, there is a bag $\beta(t)$ containing $x$ by
\ref{enum:vertex-treedec}. We may pick any such bag for each vertex $x \in S$
and obtain a set of at most $c$ bags $\mathcal{B}$.

Consider a fixed bag $\beta(t)$, and assume towards a contradiction that
$w(t,U_1)>2c$ and $w(t,U_2)>2c$. For each $U_i$, we deduce a set $X_i \subseteq
\beta(t)$ such that $|X_i|>c$ and for which there is a set of $|X_i|$
vertex-disjoint paths $\mathcal{P}_i$ from $X_i$ to $U_i$. Indeed, if $|\beta(t)
\cap U_i|>c$, then we choose trivial paths on $\beta(t) \cap U_i$. Otherwise,
there is a minimal set of more than $c$ adhesions hitting all $\beta(t)$-$(U_i
\setminus \beta(t))$ paths. By minimality, each adhesion has a private vertex.
Let $X_i$ be the set of such private vertices. We deduce from
\ref{enum:conn-treedec} applied to each adhesion the set $\mathcal{P}_i$ of
$|X_i|$ vertex-disjoint $X_i$-$U_i$ paths. Finally, since $|X_1|,|X_2|>c$, at
least one path in each of $\mathcal{P}_1,\mathcal{P}_2$ is not hit by $S$. Let
$x_1 \in X_1,x_2\in X_2$ be the endpoints in $\beta(t)$ of such paths.
By an application of \ref{enum:k-lean-menger}, $x_1$ and $x_2$ must be in the
same connected component of $G-S$ since there is no separation of $B$ of order
at most $k\geq c$. We conclude that $U_1$ and $U_2$ are not separated by $S$, a
contradiction.
\end{proof}

By applying the lemmata simultaneously for all pairs of sets in $\{(N[u],N[v])
: u,v \in V(G) \}$, we obtain a parameter equivalent to the largest dipole
topological minor based on observing simple properties of well-formed $k$-lean
tree decompositions.

\begin{theorem}\label{thm:dec-dipole}
If there is no $D_k$ topological minor in $G$, then there exists $(a,b,c)=f_1(k)$
such that, for any $a' \geq a$, for any well-formed $a'$-lean tree decomposition
$(T,\beta)$ of $G$, each bag $\beta(t)$ of $T$ has at most one vertex $v$ with $w(t,N(v))
> b$, and, for each pair of vertices $u,v \in V(G)$, there is a set $\mathcal{B}_{u,v}$ of
at most $c$ bags of $T$ such that $\beta(\mathcal{B}_{u,v})$ hits all $N(u)$-$N(v)$ paths.

Conversely, if there exists $(a,b,c)$ with the properties above, then there is no $D_{k'}$
topological minor for $k'=f_2(a,b,c)$.
\end{theorem}

\begin{proof}
\begin{description}[listparindent=1.5em]
\setlength{\itemsep}{7pt}
\item[($\Rightarrow$)] By \cref{obs:dipole-to-sep}, for all pairs of vertices $u,v \in
V(G),$ $\mu(N(u),N(v))$ is bounded by $k'$, itself bounded by a function of $k$. Let
$(T,\beta)$ be a well-formed $a'$-lean tree decomposition, with $a' \geq k'$. We apply
\cref{lem:sep-to-bound} to $(T,\beta)$ and pairs of sets $(N(u),N(v))$ for $u,v \in V(G)$.
In particular, for a fixed $t \in V(T),$ and $u,v \in \beta(t)$, only one of $w(t,N(u))$
and $w(t,N(v))$ is greater than $b$, with $b$ bounded by a function of $k$. We deduce that
there is at most one vertex $v$ of $\beta(t)$ such that $w(t,N(v))>b$. The existence the
desired sets of at most $c$ bags $\mathcal{B}_{u,v}$ follows from the application of
\cref{lem:sep-to-bound} to $(N(u),N(v))$.

\item[($\Leftarrow$)] Let $(T,\beta)$ be a well-formed $a$-lean tree decomposition with
the given properties. By \cref{lem:bound-to-sep} and \cref{obs:dipole-to-sep}, the dipole
number is bounded by a function of $a,b,$ and $c$.
\qedhere
\end{description}
\end{proof}

\section{Folding decompositions}\label{sec:folding}

We will now apply a strategy we call ``folding'' a tree decomposition, which is
inspired by constructions for tree decompositions of logarithmic depth
\cite{contractingtreedec} and constructions of tree-partitions or domino tree
decompositions
\cite{tpwDegree,DominoTw1,DominoTw2}. Given some initial tree decomposition
$T$, we will recursively consider an adhesion of bounded size for which we want
to bring the neighbours of its vertices closer in the decomposition. The main
goal is to do so without ruining the properties of the initial decomposition.
Intuitively, this is done by building a new tree decomposition with slightly
larger adhesions on top of a ``folded'' $T$, see Figure~\ref{fig:folding}. 

\begin{figure}[h]
\centering
\includegraphics{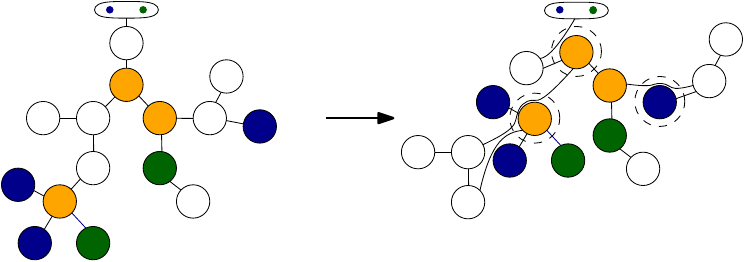}
\caption{An illustration of folding for the case of the fan. We consider a
fixed adhesion containing a blue vertex and a green vertex, and the subtree of
the tree decomposition below this adhesion. In blue (resp. green), the bags
introducing neighbours of the blue (resp. green) vertex. In orange, the branching
nodes of the set of nodes coloured in blue or green. Our construction contracts
subtrees of uncoloured nodes and pushes them away from the adhesion. Dashed
circles correspond to new bags which contain the adhesions that are pushed
down.}
\label{fig:folding}
\end{figure}

The previous structure theorems can be interpreted as follows. First, each bag
satisfies constraints that bound the size of a model of the forbidden structure
that would be centered on it. Second, we have sets of bags hitting all large
models of the forbidden structure, for which we either bound the ``topological
diameter'' for the fan or simply the size for the dipole. 

While the previous structure theorems could be abstracted in terms of arbitrary
sets of vertices, we will now rely on the fact that we consider neighbourhoods
and that the bags containing a vertex span its neighbourhood. In particular,
the number of neighbourhoods crossing an adhesion of $T$ is bounded by the
number of vertices of the adhesion.

Combining the properties raised in the above two paragraphs, we can bound the
complexity of the folding below a fixed adhesion. The folding pulls a set of
bags closer to an adhesion and is applied recursively to the subtrees of bags
that were not pulled. An important observation is that the size of adhesions
in recursive calls remains bounded. We now state the structure theorems obtained after the
folding.

\begin{theorem}\label{thm:folded-fan}
If there is no $F_k$ topological minor in $G$, then there exists a tree
decomposition $(T,\beta)$ of $G$ such that adhesion size is at most $a$, each vertex
$v$ has at most $b$ neighbours in the torso at each bag, and the bags containing $v$
induce a subtree of $T$ of diameter at most $c$.

Conversely, if such a tree decomposition exists, there is no $F_{k'}$ topological minor in
$G$ for some $k'=f(a,b,c)$.
\end{theorem}

\begin{proof}
First, we show that the properties on the tree decomposition $(T,\beta)$ in the statement
imply a bound on the fan number. Without loss of generality, we may assume $(T,\beta)$
satisfies \ref{enum:unique-adh}. Then $(T,\beta)$ satisfies the assumptions of
\cref{lem:bound-to-td} on $N[v]$ for each $v \in V(G)$. Indeed,
since adhesions incident to $t \in V(T)$ with $v \in \beta(t)$ leading to neighbours of
$v$ must contain $v$ (properties \ref{enum:vertex-treedec} and \ref{enum:edge-treedec}),
have size bounded by $a$, and are contained in $N_{G_t}[v]$, either $\beta(t)$ is of
bounded size, or the adhesions are distinct by \ref{enum:unique-adh} so we may bound their
number. We deduce a bound on the treedepth of $N[v]$ for each $v \in V(G)$, and apply
\cref{thm:neighbourhood-td} to conclude.

We now show how to construct a decomposition with the desired properties.
We start from a decomposition $T^0$ as given by \cref{thm:dec-fan}.
We initially pick any adhesion $A_0$ of $T^0$.
Then we recursively reorganise $T^0$ into a decomposition $T$ with the desired
property as follows:
We always consider a set $A$ of vertices and a subtree $T'$ of $T^0$ with up to $2$ adhesions separating
it from the rest of $T^0$ and such that no vertex has been placed in a
bag of $T$ yet except for the vertices in $A$, which have bounded
span in $T$ and such that there exists a bag containing the entire set $A$. Moreover, $A$ is a subset of these $2$ adhesions.  

The pair $(T',A)$ is processed as follows.
Let $\mathcal{B}$ be the set of bags obtained as follows. We fix an arbitrary root $r$ of
$T'$ and let $\mathcal{B}_1$ denote the set of bags containing the introduction of every
edge incident to $A$. Let $T_1$ denote the subtree of $T'$ obtained by taking for every
element of $\mathcal{B}_1$ the unique path to the root. Let $\mathcal{B}$ denote the union
of $\mathcal{B}_1$ and the set $\mathcal{B}_2$ of bags of degree greater than $2$ in
$T_1$. If we denote by $\widehat{T}$ the tree obtained from $T_1$ by contracting vertices
of degree $2$ (so keeping only $\mathcal{B}$), we can show that $\widehat{T}$ has height
$2ac$. Indeed, if we fix any root-to-leaf path of $T_2$, then any vertex on that path
corresponds either to the introduction of an edge incident to a vertex of $A$ or to a
branch leading to a bag containing such an edge. Since for every element of $A$ this
corresponds to at most $c$ bags by the second property of \cref{thm:dec-fan} and $A$ has
size at most $2a$, the total length of the path is bounded by $2ac$. For any adjacent pair
of bags $B_1, B_2$ of $\widehat{T}$, let $T_{B_1, B_2}$ denote the connected component of
$T' \setminus  \mathcal{B}$ whose neighbourhood in $T'$ is contained in $B_1 \cup B_2$.
Note that by definition of $\mathcal{B}$, this covers all the connected components of $T'
\setminus  \mathcal{B}$ whose neighbourhood is contained in more than one bag of
$\mathcal{B}$. We will now add $A$ to every bag of $\widehat{T}$ and add $\widehat{T}$
below the bag of the part already processed that contains $A$. Moreover, for every
adjacent pair of bags $B_1, B_2$ of $\widehat{T}$ such that $B_1$ is a child of $B_2$,
then we add to $B_1$ all the elements of the adhesions of $T_{B_1, B_2}$ with $B_1$ and
$B_2$. Then, for each $T_{B_1, B_2}$ we remove $A$ from its bags, and recurse on it with
the union of its adhesions to $B_1$ and $B_2$ (from which we remove $A$). For other
subtrees $T_B$ of $T' \setminus \mathcal{B}$, their adhesion to $\widehat{T}$ is contained
in a single bag $B$. We recurse on $T_B$ with its adhesion to $B$. We do not need to
remove $A$ from the bags in $T_B$ by \ref{enum:adh-neighbour} and the definition of
$\mathcal{B}$. Since we covered the cases of all the connected components of $T' \setminus
\mathcal{B},$ we will indeed process the entire subtree $T'$. 

The fact that we remove $A$ when recursing is crucial to the height of elements containing
a single vertex $x$ at the end. Indeed, the construction ensures that every vertex only appears at two levels of recursion, and since at each level it spans a subtree of diameter $2ca$, we have a total diameter of $6ca$. The reason why we can remove $A$ when recursing while keeping a proper tree decomposition comes from the fact that $\mathcal{B}$ contains a bag per edge incident to $A$ and thus already the part of the decomposition that corresponds to $\widehat{T}$ contains all these edges. 
\end{proof}

From the construction, we immediately deduce the following corollary.

\begin{observation}\label{obs:folded-bag}
A bag of $\widehat{T}$ (or $T$) consists of a subset of the union of a bag of $T^0$ and
at most three\footnote{We can slightly reduce the size of bags by adding only one of the adhesions
in $A$ to bags of $\widehat{T}$ and rooting $T'$ at its node incident to the other
adhesion. We then have bags of $\widehat{T}$ that are subsets of a bag and two adhesions
of $T^0$.} adhesions of $T^0$.
\end{observation}

We can slightly modify our construction to obtain the following statement. We do not
discuss the running time required to implement this.
\begin{observation}
We can ensure that $T$ has diameter $O(f(k)\log(n))$ as well as the properties of
\cref{thm:folded-fan}. 
\end{observation}

\begin{proof}
We reuse the notations of the proof of \cref{thm:folded-fan}.
To ensure that $T$ will have diameter $O(f(k)\log(n))$, we can simply add to $\mathcal{B}_1$ a
balanced separator $B'$ of $T'$.
This implies each subtree on which we recurse has size reduced by a constant multiplicative
factor. This trick increases the diameter of $\widehat{T}$ by a constant additive factor
while ensuring that the recursion depth is $O(\log(n))$.
\end{proof}

We also apply the folding technique to graphs excluding a dipole as a topological
minor to obtain the following.

\begin{theorem}\label{thm:folded-dipole}
If there is no $D_k$ topological minor in $G$, then there exists a tree
decomposition $(T,\beta)$ of $G$ such that adhesion size is at most $a$, each bag
$\beta(t)$ has at most one vertex $v$ with more than $b$ neighbours in $\beta(t)$ or
more than $c$ adhesions containing $v$ incident to $t$, and for each pair of
vertices $u,v$ there are at most $d$ bags $B$ such that $\{u,v\}\subseteq B$.

Conversely, if such a tree decomposition exists, there is no $D_{k'}$ topological minor in
$G$ for some $k'=f(a,b,c,d)$.
\end{theorem}

\begin{proof}
First, we show that the properties on the tree decomposition $(T,\beta)$ in the statement
imply a bound on the dipole number. Without loss of generality, we may assume $(T,\beta)$
satisfies \ref{enum:unique-adh} and \ref{enum:adh-neighbour}. Then $(T,\beta)$ satisfies
the properties of \cref{lem:bound-to-sep} on $N(u),N(v)$ for each pair $u,v \in V(G)$. Indeed,
for all $v \in \beta(t)$ except one, the adhesions incident to $t \in V(T)$ leading to
neighbours of $v$ must contain $v$ (properties \ref{enum:vertex-treedec} and
\ref{enum:edge-treedec}), but there are at most $c$ such adhesions.
Together with the bound on neighbours of $v$ in $\beta(t),$ this bounds $w(t,N(v))$.
Observe that the set of bags containing $\{u,v\}$ hits all $N(u)$-$N(v)$ paths and
contains at most $d$ bags. We deduce that for all pairs $u,v \in V(G),$ $\mu(N(u),N(v))$
is bounded by a function of $a,b,c,d$, and conclude with \cref{obs:dipole-to-sep}.

We now show how to construct a decomposition with the desired properties, starting from
$T^0$ a well-formed $a$-lean tree decomposition guaranteed by \cref{thm:dec-dipole}.
To obtain the final decomposition $T$, we fold $T^0$ by recursively separating
the neighbours of the adhesions incident to a subtree. We start by picking an
arbitrary adhesion $A$ of $T^0$, and start folding the subtrees $T_i$ of $T$
separated by $A$.

We now assume that we are given a subtree $T'$ of $T^0$ that is separated from
the rest of $T^0$ by at most $2$ adhesions, and that the only vertices in
$\beta(T')$ that are already placed in $T$ are the vertices in these adhesions.
These adhesions amount to a set of vertices $X$ of at most $2a$ vertices,
for each pair of vertices $u,v \in X$, we can obtain a set of $c$ bags
that separate $N(u)$ from $N(v)$ in $T'$ using \cref{lem:sep-to-bound}. 
Observe that the total number of bags is bounded by a function of $a$ and $c$. 
Let $\mathcal{B}_X$
denote this set of bags. We obtain a tree $\widehat{T}$ for $\mathcal{B}_X$ by
adding bags corresponding to branching nodes and connecting bags that are connected by
direct paths in $\mathcal{T^0}$. We root $\widehat{T}$ arbitrarily, then we add $X$ to the
bags of $\mathcal{B}_X$, and for each pair of bags adjacent in $\widehat{T}$ that
are not adjacent in $T^0$, we add the missing adhesion to $T_i$ in the deepest
node of $\widehat{T}$ adjacent to $T_i$ in $T^0$ (similarly to
\cref{thm:folded-fan}). Let $T_i$ be the subtrees of $T'$ separated by
$\mathcal{B}_X$, for each of them, we may remove all vertices of $X$ except
one. We then recurse on each $T_i$.

In $T$, the number of bags containing a pair of vertices is bounded by a function of $c,a$
and $k$. Indeed, we can bound the number of subtrees on which we recurse that may contain
a fixed pair. There can be $O(ca^2)$ subtrees of $T^0$ that are incident to two
adhesions of $T^0$, and there can be at most $k$ subtrees incident to a single adhesion
of $T^0$ by \ref{enum:conn-treedec}. Hence, a pair of vertices will always appear only in
a first $\widehat{T}$ and then in the subtrees $\widehat{T}$ of a bounded number of
recursions. We can conclude from the bound on the size of $\widehat{T}$. Other properties
of $T$ are inherited from $T^0$ and the fact that we only add a constant amount of
adhesions to each bag due to the folding.
\end{proof}

We make the observation that in the case of domino treewidth, it is possible to
have each vertex in at most two bags. This leads to the following question.
\begin{opquestion}
Can we always achieve $d \leq 2$ in the statement of \cref{thm:folded-dipole}? What if we also bound the treewidth of $G$?
\end{opquestion}

\section{Implementation}

We detail how to implement the constructions of \cref{thm:folded-fan} and
\cref{thm:folded-dipole} in this section.

\subsection{Excluded fan}\label{subsec:implem-fan}

In this subsection, we discuss how to efficiently check that the structure
theorem for a forbidden $k$-fan holds and then compute a folded tree
decomposition $\widetilde{T}$. It follows from known results \cite{Bollobas1998,Komlos1996}, that if $G$
excludes a $k$-fan, then $m = O(k^{2}\cdot n)$. We may safely conclude that
a graph $G$ not satisfying this bound does not exclude the $k$-fan as a
topological minor. We now assume that this bound holds.

The first ingredient is the efficient computation of a $k'$-lean tree
decomposition by Korhonen \cite{KorhonenkLean}. We then produce a well-formed
$k'$-lean tree decomposition with a standard linear time processing. We now
consider a fixed instance with a value $k$ and a graph $G$. We define $k'=
\binom{k}{2}+1$ following \cref{lem:fan-to-td}.

\begin{theorem}
There is an algorithm that computes a $k'$-lean tree decomposition of
$G$ in time $k^{O(k^4)}\cdot n$.
\end{theorem}

Let $T$ be a well-formed $k'$-lean tree decomposition.
We assume $T$ to be rooted, the root can be chosen arbitrarily. 

We will first describe how to compute an auxiliary tree $T^x$ for every vertex
$x\in V(G)$, such that nodes of $T^x$ are a subset of nodes of $T$ whose bags
cover $N[x]$ with the addition of their branching nodes in $T$. Then for each
adhesion of $T$, only the auxiliary trees $T^x$ for $x$ in the adhesion have
nodes on both sides of the adhesion. Once these trees are computed, it is
possible to efficiently compute the subtree $\widehat{T}$ from
\cref{thm:folded-fan}. Indeed, it can be computed by combining relevant parts
of subtrees $T^x$. This will require some subroutines to compute relevant nodes in $T$ in
constant time.

\begin{lemma}\label{lem:data-structures}
Given a rooted tree $T$, there is a data structure which can be computed in time
$O(|V(T)|)$, and can answer the following queries in time $O(1)$: 
\begin{itemize}
\item for $a \in V(T)$, give the depth of $a$ in $T$.
\item for $a,b \in V(T)$, check if $a$ and $b$ are in ancestor-descendant relationship. 
\item for $a,b \in V(T)$, give the lowest common ancestor of $a$ and $b$.
\item for $a,b,c \in V(T)$, give the branching node of $a,b,c$.
\end{itemize}
\end{lemma}

\begin{proof}
By performing a DFS from the root of $T$, we can store three values per node
$u$ of $T$ corresponding to the depth $d[u]$, the time of discovery $\tb[u]$,
and the time of end of exploration $\te[u]$. Using these values, one can test
if two nodes are comparable (for the ancestor-descendant relation) in time
$O(1)$. Indeed, a node $u$ is an ancestor of a node $v$ if and only if the
interval $[\tb[u],\te[u]]$ contains the interval $[\tb[u],\te[u]]$. The DFS can
be performed in time $O(|V(T)|)$. The exploration order also provides a
canonical ordering on $V(T)$, and nodes can be compared with respect to this
ordering in time $O(1)$.

We also use a data structure to compute lowest common ancestors of pairs of
nodes in time $O(1)$. Such a data structure can be computed in time
$O(|V(T)|)$, see e.g. \cite{Alstrup2004}.

Now given a triple of nodes $a,b,c$, the branching node is always the furthest
from the root among lowest common ancestors of pairs of $a,b,c$.
\end{proof}

For each vertex $x \in V(G)$, we define $t(x)$ to be the node $t$ of $T$ such
that $x \in \beta(t)$ which is closest to the root.
For each edge $xy$ of $G$, we map it to a bag $t(xy)$ which is the furthest
from the root among $t(x)$ and $t(y)$. This is well defined since, by
\ref{enum:vertex-treedec} and \ref{enum:edge-treedec}, the subtrees induced by
bags containing $x$ and $y$ intersect. In particular, $t(x)$ and $t(y)$ are
comparable. We can compute $t(x)$ for every $x\in V(G)$, and $t(xy)$ for every
edge $xy \in E(G)$ in time $k^{O(1)}\cdot n$. We store the computed values in tables, and,
for each $t \in V(T)$, store a list $I[t]$ of the edges $xy$ introduced by $t$,
i.e. such that $t=t(xy)$.

We represent the trees $T^x$ with a pointer data structure where each node of
$T^x$ has a list of pointers to its incident edges and edges have pointers to
their incident vertices. Each adhesion of $T$ will store
a pointer to the edge of any $T^x$ that crosses it, and we also store in a global table
the pointers to the root of $T^x$ for each $x \in V(G)$. Note that adhesions $\alpha(st)$
have size at most $k'$, and that a tree $T^x$ will cross $uv$ if and only if $x \in
\alpha(st)$.

\begin{lemma}
In time $k^{O(1)}\cdot n$, we can compute the trees $T^x$ and check the
conditions of \cref{thm:dec-fan}. 
\end{lemma}

\begin{proof}
We proceed bottom-up on $T$. For each node $t$ of $T$, we will compute: the
subtrees $T_t^x$ of the $T^x$ induced by the nodes of the subtree of $T$ rooted at $t$;
and a list containing, for each vertex $x$ in the adhesion $\alpha(tp)$ between $t$ and
its parent $p$, a pointer to the root of $T_t^x$. 

Let $t$ be a node of $T$, first compute the subtrees of $T^x$ for each of its
children in $T$. Then, we have two dictionaries $D,C$ indexed by vertices of
$\beta(t)$, with all entries of $C$ initially storing value $0$, and entries of
$D$ which are initially empty. We first iterate on the lists of pointers given
by each adhesion between $t$ and its children $s_i$. For each pointer \texttt{p}
pointing to the root of subtree $T^x_{s_i}$ of $T^x$, we increment $C[x]$ and update the
entry $D[x]$ with the following cases:
\begin{itemize}
\item if $D[x]$ was empty, we store \texttt{p} in $D[x]$;
\item if $D[x]$ contains a pointer \texttt{p'} to a representative of node $s
\neq t$, create a representative of node $t$ for $T^x$, link it to nodes pointed
by \texttt{p} and \texttt{p'}, and store a pointer to the representative of $t$
in $D[x]$; and
\item if $D[x]$ contains a representative of node $t$, link it to the node
pointed by \texttt{p}.
\end{itemize}

We then iterate on the edges in $I[t]$, and, for each edge $xy$, we increment
$C[x]$ and $C[y]$, and if $D[x]$ or $D[y]$ does not already point to a
representative of $t$, we create such a representative, and link to the
previous representative that was pointed to (if there is one).

To check that condition 1 of \cref{thm:dec-fan} hold, it suffices
to check that either $\beta(t)$ has bounded size, or that $C[x]$ is bounded for every $x
\in \beta(t)$.

The list of pointers for the adhesion between $t$ and $p$ simply consists of
the values of $D[x]$ for each vertex $x \in \alpha(tp)$.

Once the procedure described above has been applied to every node $t$ in $T$,
the trees $T^x$ are computed. We compute their diameter to check the
condition 2 of \cref{thm:dec-fan}. We also update the pointers to the trees
$T^x$ stored in the adhesions of $T$ so that they point to edges of $T^x$.

All the above computation can be performed in time $k^{O(1)}\cdot n$. Indeed,
$\sum_{x \in V(G)} |V(T^x)| = O(n + m)$, the rest of the data structures are of
size proportional to bags and adhesions which have total size $k^{O(1)}\cdot
n$, and are accessed $O(n + m)$ times in constant time.
\end{proof}

We move to the implementation of the folding procedure. If the above procedure
found that the conditions of \cref{thm:dec-fan} are not satisfied, we can
reject. We now assume that they are satisfied.

\begin{lemma}
The construction of \cref{thm:folded-fan} can be implemented in time $k^{O(1)}\cdot n$.
\end{lemma}

\begin{proof}
We describe the construction of $\widetilde{T}$, the folded decomposition obtained from $T$.
The steps of the construction are recursively described by pointing to at most two adhesions $e,e'$ of $T$. 
If two adhesions are given, we must fold the subtree $T'$ of $T-\{e,e'\}$ that is incident to $e$ and $e'$.
Otherwise, if only the adhesion $e \in E(T)$ is given, we must fold the subtree $T'$
of $T-e$ that does not contain the root of $T$. We also maintain the set of
vertices $X$ that appear in the adhesions of $T$ but have been already placed
in $\widetilde{T}$ with all their neighbours. Following the notations of
\cref{thm:folded-fan}, let $A$ denote the vertices in the adhesions
that separate $T'$ from the rest of $T$. By \ref{enum:vertex-treedec}, $X$ is a
subset of $A$, and, by definition of $k'$, we have $|X| \leq |A| \leq k^{O(1)}$.

From $e,e'$ we get pointers to edges of $T^x$ that cross $e,e'$ for vertices $x
\in A \setminus X$. Let $\mathcal{T}$ denote this set of trees.
The tree $\widehat{T}$ from \cref{thm:folded-fan} can be constructed as the
tree on the union of nodes of $T$ that have representatives in the trees of
$\mathcal{T}$ and their branching nodes. The tree $\widehat{T}$ is computed
recursively in time $k^{O(1)}\left(\sum_{T^x \in \mathcal{T}} |T^x|\right)$ as
follows. 

\subparagraph{Computing $\widehat{T}$.} We explore in parallel the subtrees induced by $T'$ in $\mathcal{T}$,
i.e. we only explore the representatives of nodes $t$ of $T^x \in \mathcal{T}$
such that $t \in V(T')$. For each subtree in $\mathcal{T}$, we start the
exploration from their node of $T'$ closest to the root. These nodes are obtained either
via the pointers stored at the adhesions or from the table of pointers to the roots of
each $T^x$.
At each step, we compute the node $t$ of $T$ that is lowest common ancestor of the
currently considered nodes in $\mathcal{T}$. This node $t$ is
chosen as the root of the subtree of $\widehat{T}$. Then, we move the currently
considered node to its children in trees of $\mathcal{T}$ for which the chosen
root was the current node. This gives a set of nodes $W$ in $\mathcal{T}$ which we
partition based on the subtree of $T'$ rooted at $t$ in which they belong. We
can assume that nodes from each tree of $\mathcal{T}$ are already given following
the canonical order on $V(T)$.

To compute the partition, we first merge the ordered lists of nodes from each
tree of $\mathcal{T}$ into an ordered list of $S$ which follows the canonical
order on $V(T)$. This takes time $k^{O(1)}\cdot|S|$. From the ordered list of
$W$, the partition can be deduced from lowest common ancestor queries between
nodes that are consecutive in the ordered list of $S$. Indeed, since all nodes
of $S$ are descendants of $t$, two nodes of $S$ are in the same subtree if and
only if their lowest common ancestor is not $t$. Furthermore, due to their
ordering, elements in a part are already consecutive. We recursively compute a
subtree of $\widehat{T}$ for each part of the partition of $S$.

Once $\widehat{T}$ is computed, we add $A \setminus X$ and remove $X$ from its bags, and then recurse
on the subtrees $T_i$ of $T'$ that are between nodes of $\widehat{T}$. Let
$A_i$ be the set of vertices in adhesions incident to $T_i$. We set $X_i$ to be
$(X\cup A) \cap A_i,$ the vertices that have to be removed from bags of $T_i$.

The overall computation takes time $k^{O(1)}\cdot|V(T)| + k^{O(1)} \cdot n = k^{O(1)} \cdot
n$. Indeed, the computation of the trees $\widehat{T}$ takes only time
$k^{O(1)}\cdot |V(T)|$ because they correspond to disjoint subtrees of $T$.
Updating the vertices in bags takes time $k^{O(1)} \cdot n$.
\end{proof}

We conclude with the following theorem.

\begin{theorem}
There is an algorithm, running in time $k^{O(k^4)}\cdot n$, which either certifies that
$G$ contains a $k$-fan, or produces the tree decomposition from
\cref{thm:folded-fan}.
\end{theorem}

\subsection{Excluded dipole}\label{subsec:implem-dipole}

In this subsection, we discuss how to efficiently check that the structure
theorem for a forbidden $k$-dipole holds and then compute a folded tree
decomposition. Most of the techniques are introduced in \cref{subsec:implem-fan}, we
describe briefly modifications to cover the case of forbidden dipoles. We may
still apply the results of \cite{Bollobas1998,Komlos1996} to bound $m$ by
$O(k^{2}\cdot n)$.

In this case, we first obtain a well-formed $(k-1)$-lean tree decomposition $T$ using the
algorithm of Korhonen \cite{KorhonenkLean}. We then compute trees $T^{xy}$ for each pair
$\{x,y\}$ of vertices that share more than one bag. $T^{xy}$ must consist of bags of $T$ that
hit all $x$-$y$ paths of $G$. We build the trees $T^{xy}$ similarly to our construction of
the trees $T^{x}$ in the case of the construction for forbidden $k$-fans. However, in the
case of dipoles, we also store the following information for each subtree $T'$ of $T$.
Let $T''$ be the subtree of $T$ rooted at the root of the partial $T^{xy}$ induced by
$T'$, we store at the root of $T'$, whether $T'-T''$ has introduced a neighbour of $x$ or
$y$. We add the root of $T'$ as a bag of $T^{xy}$ if, either it is a branching node (i.e.
at least two subtrees have partial subtrees of $T^{xy}$), or both $x$ and $y$ have
introduced a neighbour in $T'-T''$.

To check the conditions of the structure theorem, we check in each bag that each vertex,
except at most one, has bounded degree in the bag and appears in a bounded number of
adhesions incident to the bag. Note that it suffices to consider the degree in terms of
edges introduced in the bag since, for each other edge $xy$, there is an adhesion, containing
both $x$ and $y$, incident to the bag. We also check that each computed tree $T^{xy}$ has
bounded size. We stress that for pairs $\{x,y\}$ with no computed tree $T^{xy}$, there is at
most one bag containing both vertices, so they also satisfy the conditions.

The implementation of the folding does not change much except that now $\widehat{T}$ is
computed by combining the trees $T^{xy}$ for each pair $x,y$ contained in the adhesions.
In particular, some trees $T^{xy}$ may not cross the adhesions. Nonetheless, we can check
that they lie in the considered subtree $T'$ in constant time by checking that any of their
nodes lies in $T'$. Furthermore, the vertex $x$ of the adhesions incident to $T'$ which has
neighbours in the subtree on which we recurse has been precomputed as part of the
computation of $T^{xy}$. There are at most $k^{O(1)}$ pairs of vertices contained in the two
adhesions incident to $T'$, so we have to combine at most $k^{O(1)}$ trees
(of size $k^{O(1)}$) to obtain $\widehat{T}$ which means the running time of the folding
will still be $k^{O(1)}\cdot n$.

\begin{theorem}
There is an algorithm, running in time $k^{O(k^2)}\cdot n$, which either certifies that
$G$ contains a $k$-dipole, or produces the tree decomposition from \cref{thm:folded-dipole}.
\end{theorem}

\section{Corollaries}\label{sec:coro}

\begin{theorem}\label{thm:tbw-algo}
	There is an algorithm that given a graph $G$ and an integer $k$, in time
	$k^{O(k^2)}\cdot n$ either computes a tree-layout of bandwidth at most $g(k)$, or
	determines that $\tbw(G)>k$.
\end{theorem}

\begin{proof}
We will conclude that $\tbw(G)>k$ if we reach the conclusion that $G$ contains
$F_{2^k}$, or if $\tw(G)>k$.

If $G$ excludes $F_{2^k}$ as a topological minor, then $\td(G,N[v])\leq 2^{2k}$
by \cref{lem:fan-to-td}. We may compute a $k$-lean tree decomposition $T$ of
$G$ in time $k^{O(k^2)}\cdot n$ using Korhonen's algorithm
\cite{KorhonenkLean}. If $G$ does not have width $k$, we may reject since
$\tw(G)>k$. Otherwise, $T$ is lean and in particular it is $k'$-lean for any $k'>k$.

Without loss of generality, we may assume that $T$ is well-formed. We then
proceed with the implementation of \cref{thm:folded-fan} using $T$ as our
decomposition since it is a well-formed $k'$-lean tree decomposition, with
$k'=2^{2k}$. More precisely, we check that the conclusions of
\cref{lem:td-to-bound} hold and, if they do, compute a folded decomposition
$T'$ from $T$. Otherwise, we get a contradiction with our assumption that $G$
excludes $F_{2^k}$ as a topological minor.

We conclude by constructing a tree-layout $T''$ from $T'$ in the usual way:
root it arbitrarily, replace each bag by any linear order, and keep only the
occurence of a vertex that is closest to the root. Bags in $T'$ have size $O(k)$, and the
subtree of $T'$ containing each vertex has diameter at most bounded by a function of $k$.
We deduce that $T''$ has bandwidth at most $g(k)$ for some function $g$ \footnote{The
function $g$ is doubly exponential with the current proof.}.
\end{proof}

We did not attempt to optimise the quality of the approximation, and leave the following question open.
\begin{opquestion}
Is there an FPT algorithm which, given a graph $G$ and an integer $k$, either
determines that $\tbw(G) > k$, or computes a tree-layout of bandwidth
$k^{O(1)}$ in time $f(k)\cdot n^{O(1)}$.
\end{opquestion}

A positive answer to this questions seems reasonable but would require to fold
fans using their logarithmic depth tree-layout or detect an obstruction to such
a tree-layout. In this direction, one could try to adapt arguments from
\cite{CzerwinskiNP21} to the rooted minor setting. Moreover, it should be
possible to combine paths model instead of looking for a complete obstruction
at each inductive step in the proof of \cref{lem:td-to-bound}.

\begin{definition}[overlap treewidth]
Given a tree decomposition $(T,\beta)$ of a graph $G$, we define its \emph{overlap number}
to be the maximum over pairs $u,v \in V(G)$ of the number of bags $\beta(t)$ such that
$u,v \in \beta(t)$. The \emph{overlap treewidth} of $(T,\beta)$ is the maximum
of its width and overlap number.
We define the \emph{overlap treewidth} of a graph $G$, $\otw(G)$, to be the minimum overlap treewidth
over tree decompositions of $G$. 
\end{definition}

The following observation follows from \cref{thm:folded-dipole}.

\begin{observation}
Dipoles and walls are the only obstructions to large overlap treewidth.
\end{observation}

\begin{theorem}
	There is an algorithm that given a graph $G$ and an integer $k$, in time
	$f(k)\cdot n$ either computes a tree decomposition of overlap treewidth at most
	$k^{O(1)}$, or determines that $\otw(G)>k$.
\end{theorem}

\begin{proof}
We will reach the conclusion that $\otw(G)>k$ if $\tw(G)>k$, or $G$ contains $D_{k^2}$ as a topological minor.

We compute a $k$-lean decomposition $T$ using Korhonen's algorithm
\cite{KorhonenkLean}, then check if its width is at most $k$. Either we conclude that
$\otw(G)>k$, or we conclude that $T$ is lean. We then make it well-formed
and apply the folding implementation on $T$ corresponding to the structure theorem for graphs excluding $D_{k^2}$. 

We either conclude that $D_{k^2}$ is a topological minor, or obtain the decomposition from
\cref{thm:folded-dipole} with the additional property that bags have size $O(k)$. In
particular, this decomposition has overlap number bounded by a polynomial function of $k$.
\end{proof}

\section{Centered colourings}\label{sec:colouring}

The link between treebandwidth and treedepth can be used to easily derive efficient
algorithms based on treedepth. We first give the definition of a $p$-centered colouring
which is a central notion in the theory of sparsity and its algorithmic applications.

\begin{definition}
A \emph{$p$-centered colouring} of $G$ is a colouring $\lambda:V(G) \to \mathbb{N}$ such
that for every connected subgraph $H$ of $G$, either $f$ uses more than $p$ colours on $H$
(i.e. $|\lambda(V(H))|>p$), or there is a colour assigned to a unique vertex in $H$.
\end{definition}

It is already known that graphs of treewidth $k$ have $p$-centered colourings with
$O(p^k)$ colours \cite{PilipczukSiebertz}, and that graphs avoiding a fixed topological
minor $H$ have $p$-centered colourings with $O(p^{f(H)})$ colours\cite{pCenteredTopo}.
We can significantly improve over these bounds when restricting ourselves to the case of
graphs of bounded treebandwidth.

\begin{theorem}
If a graph $G$ has treebandwidth at most $k$, it has $p$-centered colourings with $pk+1$
colours.
\end{theorem}

\begin{proof}
Let $T$ be a tree-layout of $G$ of bandwidth $k$. Consider the colouring $\lambda$ defined
by assigning the depth of $x$ in $T$ modulo $(pk+1)$ to $\lambda(x)$ for every $x \in G$.

Let $H$ be a connected subgraph of $G$, and $x$ be the vertex of $H$ that is closest to
the root in $T$. If some other vertex $y$ of $H$ is assigned the same color as $x$, then,
by definition of $\lambda$, they must be at distance at least $pk+1$ in $T$. Using the
bandwidth bound, this implies that the shortest path between $x$ and $y$ in $H$ covers at
least $p$ vertices that are strictly between $x$ and $y$ in $T$ and that are at distance
at most $pk$ from $x$ in $T$. By construction of $\lambda$, we conclude that these
vertices have distinct colours, implying that $H$ is coloured with more than $p$ colours.
\end{proof}

We give some explicit algorithmic applications of $p$-centered colourings to
motivate our bound. We give the example of the algorithm of Pilipczuk and
Siebertz~\cite{PilipczukSiebertz} for \textsc{Subgraph Isomorphism} (finding if $H$ is a
subgraph of $G$) that runs in time $2^{O(p \log d)}\cdot n^{O(1)}$ and space $n^{O(1)}$,
where $n,p$ are the number of vertices of the graphs $G$ and $H$ respectively, and $d$ the
depth of an elimination forest of $G$. This algorithm can be lifted to any class of
bounded expansion (in particular bounded treewidth) by first computing a $p$-centered
colouring of $G$, then guessing the colours appearing in the subgraph $H$, and searching
for $H$ in the subgraph of $G$ induced by these colours which has bounded treedepth.
Observe that the number of colours of the $p$-centered colouring directly affects the
running time of the algorithm. In our case, our improved bound directly translates in an
improved bound on the running time.

Moreover, a tree-layout of bounded bandwidth encodes good $p$-centered colourings for
every $p$. This motivates such tree-layouts as data structures to answer multiple queries.
We can compute the tree-layout of a graph $G$ once, and then efficiently answer any
subgraph isomorphism query on $G$.

\section{Hardness of computing treebandwidth}
\label{sec:hardness}

XALP was recently defined in \cite{XALP} as an extension of XNLP
\cite{XNLP0,XNLP1,XNLP2} for tree-like structures. It translates known
results about some equivalent computational models that generalise
nondeterministic Turing machines to the parameterised setting.

XNLP is the class of parameterised problems that can be solved by a
nondeterministic Turing machine that uses $f(k)\log(n)$ space and runs in
$f(k)n^{O(1)}$ (FPT) time. While the definition might seem rather technical, it
is a relatively natural class for which many hard parameterised problems are
complete when parameterised by structural parameters related to linear
structures (e.g. pathwidth). In particular, \textsc{Bandwidth} is
XNLP-complete.

XALP is (among other equivalent characterisations) the class of parameterised
problems that can be solved by a nondeterministic Turing machine that uses
$f(k)\log(n)$ space and runs in $f(k)n^{O(1)}$ (FPT) time, and additionally has
access to an auxiliary stack whose space we do not restrict. It should be clear
from this definition that XALP contains XNLP.

To prove XALP-hardness (or XNLP-hardness), one has to restrict the space
requirements of the reduction to preserve membership in the class. The least
restrictive requirement is that the reduction uses (deterministic)
$f(k)n^{O(1)}$ (FPT) time and $f(k)\log(n)$ space. One can also give a
reduction that uses $f(k) + \log(n)$ space since it also implies FPT time.

We prove that \textsc{Treebandwidth} is XALP-complete, complementing the
NP-completeness proved in \cite{spanheight}.
This has the following implications: there is a dynamic programming algorithm
running in $n^{f(k)}$ (XP) time on $n^{f(k)}$ (XP) space, and this algorithm
might very well be optimal (see \cite{SpaceEfficiency}).
If any XALP-hard problem happens to have an FPT algorithm, then the W-hierarchy
collapses (note that W[1]-hardness is usually considered a sufficient reason
why there would not be an FPT algorithm). Furthermore, this would imply
the following inclusion of classical complexity classes: SAC$^1 \subseteq$ NL.

Our proof will follow the same lines as the proof for
\textsc{Tree-Partition-Width} in \cite{algo-tpw}. We first give a sketch of the
main ideas.

\subparagraph{Membership.} We consider a nondeterministic Turing machine that stores $f(k)$
consecutive vertices of the tree-layout corresponding to a separator of the
input graph. It will progressively guess the whole tree-layout as follows. If
the graph `below' this separator is separated into several components, we will
add some information to the stack to remember that we have to check each
component. Once this is done, we progress in the current component until it has
been fully processed (i.e. there are no vertices `below'). When this happens,
we load the information on top of the stack to begin the construction of a
tree-layout for one of the connected components that were not processed yet.
Once this information is loaded, we remove it from the stack.
This describes how we guess the tree-layout, to decide if the machine accepts
or rejects, it suffices to check if the current part of the tree-layout
satisfies the distance bound.

\subparagraph{Hardness.} We reduce from \textsc{Tree-Chained Multicolour Independent
Set} \cite{XALP}. This problem can be equivalently stated as \textsc{Binary CSP}
parameterised by the width of a binary tree-partition. This problem is a
generalisation to a tree shape of the `local problem' \textsc{Multicolour
Independent Set} which asks if there is an independent set that hits every
colour of a $k$-coloured graph. We introduce it more formally above the
hardness proof. The reduction builds a graph which consists of several gadgets:
the trunk, the clique chains, and the cluster gadgets. The trunk is a part of
the graph that will enforce the global structure of the tree-layout. The
clique chains are essentially long paths attached to the trunk that are too
long for the trunk and have to be folded at their beginning and at their end to
fit in the layout.  The information about the choice of a vertex in our
independent set is encoded by the length that is folded at the beginning of a
clique chain. To check whether our choice corresponds to an independent set,
clique chains have some wider parts, and we enforce that at most one of those
wider parts can fit in some sections of the trunk specific to each edge.

\begin{lemma}
\textsc{Treebandwidth} is in XALP.
\end{lemma}

\begin{proof}
We will use the result of Reingold that undirected
connectivity is in logspace \cite{connlogspace} as a black box. We assume that
in the input graph $G$, vertices are indexed by integers from $1$ to $|V(G)|$, in
particular, we can store a vertex index with $O(\log(n))$ bits. We now describe
how to check if there is a tree-layout of $G$ of treebandwidth at most $k$ with
an algorithm respecting the specifications of XALP.

We can work on connected components of $G$ independently because given a fixed
vertex $v$ of $G$, we can check if a vertex is in the same connected component
of $G$ as $v$ in logspace. This is also sufficient to enumerate such vertices
in logspace. We call \emph{representative} of a connected component, the vertex
of the component with minimum index.

We now assume that $G$ is connected. We will denote a
\emph{configuration} of our algorithm by $(v_C,S)$ where $S$ is a linear layout
of at most $k$ vertices, $v_C$ is a vertex of a connected component of $G-S$
for which we must compute a tree-layout extending $S$. A step of our algorithm
consists of the following substeps:
\begin{enumerate}
\item Guess a vertex $u$ in the connected component $C$ of $G-S$ containing $v_C$.
\item Define $S'$ by, if $|S|=k,$ removing from $S$ its first element
$w,$, and always adding $u$ as the last element.
\item Check that every vertex of $C-\{u\}$ is separated from $w$ (if defined,
i.e. $|S|=k$) by $S'$. If not, reject.
\item Iterate on representatives $v_{C'}$ of connected components $C'$ of
$G-S'$ that contain neighbours of $u$ and do not contain $w$ (if defined, i.e.
$|S|=k$). For each of them, push $(v_{C'},S')$ on the stack. 
\item Read a configuration from the stack for the next step and pop it from the
stack. If the stack is empty, accept.
\end{enumerate}
The algorithm starts with the configuration $(x,\varnothing)$ with the
representative of $G$ (in particular it is not simply the first vertex if $G$
was not connected).

Our computation runs in $O(k\log(n))$ space. The number of steps is linear in
$|V(G)|$ and the substeps are logspace subroutines, so the algorithm runs in
nondeterministic polynomial time.

Correctness of the algorithm follows directly from the observation that, if for
each step we add vertex $u$ as a child of the last element of $S$ in a
tree-layout, and we consider a run of the algorithm that did not reject, then
the tree-layout has treebandwidth at most $k$. Indeed, a vertex $x$ of $C'$ had a
neighbour $y$ above $w$ in the tree-layout, then the algorithm would have
rejected in an earlier step since $x$ would be part of some superset of $C'$
when we removed $y$ from $S$.
\end{proof}

We now move to the reduction from the following problem.

\defparaproblem{\textsc{Tree-Chained Multicolour Independent Set}}{A (rooted) binary
tree $T=(I,F)$, an integer $k$, and for each $i \in I$, a collection of $k$
pairwise disjoint sets of vertices $V_{i,1},\ldots,V_{i,k}$, and a graph $G$
with vertex set $V=\bigcup_{i \in I, j \in [1,k]} V_{i,j}$.}{$k$.}{Is there a
set of vertices $W \subseteq V$ such that $W$ contains exactly one vertex from
each $V_{i,j}$ ($i \in I, j \in [1,k]$), and for each pair $V_{i,j},V_{i',j'}$
with $i=i'$ or $ii' \in F$, $j,j' \in [1,k]$, $(i,j)\neq(i',j')$, the vertex in
$W \cap V_{i,j}$ is non-adjacent to the vertex in $W \cap V_{i',j'}$?}

\begin{lemma}
\textsc{Treebandwidth} is XALP-hard, even when parameterised by $\tw + \Delta$.
\end{lemma}

\begin{proof}
We set $L=18k+8$. We consider an instance of \textsc{Tree-Chained Multicolour
Independent Set} (TCMIS) with notations as introduced in the problem definition
above. We will construct an auxiliary graph $H$ such that the treebandwidth of
$H$ is at most $L$ if and only if the instance of TCMIS is positive. $H$ will
also have degree bounded by a polynomial in $k$. 

We assume that each set $V_{i,j}$ is of the same cardinality $r$. This is
easily enforced by adding vertices that are adjacent to all other vertices in
$V_{i,j'}$ for all $j' \in [1,k]$ since such vertices cannot be in the
solution. Throughout the construction, we consider a fixed indexation of the
edges of $G$, $E=\{e_1,\dots,e_m\}$. We also index the elements of each $V_{i,j}$
$V_{i,j}=\{v_{i,j,1},\dots,v_{i,j,r}\}$. We set $N=2r(m+1)$, $L'=9k+3$.

We call \emph{grandparent} of a node the parent of its parent, and
\emph{grandchildren} of a node the nodes having it is their grandparent.

\subparagraph{Cluster gadgets.} Given a clique $C$ of size $c$, we call \emph{cluster
gadget} on $C$ the following construction. Add a set $C'$ of $L-c+1$ vertices
to $H$, and add edges to $H$ so that $H[C \cup C']$ induces a clique. Since
vertices of $C'$ have no other neighbours, we can always have them below
vertices of $C'$ and on a separate branch in a tree-layout of $H$.

\begin{observation}\label{obs:cluster-gadget}
If there is a cluster gadget on a clique $C$, and $H$ has treebandwidth at most
$L$, then there is a layout in which vertices of $C$ appear consecutively.
\end{observation}

\subparagraph{Trunk.} We first define a (rooted) tree $T'$ from $T$. First, we add
a grandparent to the root of $T$ (i.e. we add a parent to the root twice) and
obtain a tree $T^r$. We then subdivide each edge of $T^r$ $N$ times to
obtain $T'$. Nodes of $T$ in $T'$ are referred to as \emph{original nodes}. We
extend the use of the word original to parents and grandparents: an original
(grand)parent is a (grand)parent in $T$. For each node $i$ of $T'$, we create a
clique $A_i$ of $H$ of size $L'$ and with a cluster gadget on $A_i$, we denote
vertices of $A_i$ by $\{a_{i,1},\dots,a_{i,L'}\}$. If $i$ and $i'$ are
adjacent nodes of $T'$, we add some edges between $A_i$ and $A_{i'}$ as
described in the following.

The edges between the cliques $A_i$ and $A_{i'}$, for some $i$ with parent $i'$
in $T'$, depend on some number $p(ii')$ which counts the number of spots that
must be given to the clique chains between the two cliques. We
count the number of clique chains that `cover' $ii'$, this corresponds to $k$
multiplied by the number of original nodes below $i$ with an original
grandparent above $i'$ (where $i$ and $i'$ are respectively included). Observe
that $p(ii')$ is always in the range $[k,3k]$ because $T$ is a binary tree. For
each subdivision of an edge of $T$, we decompose it into sections of length
$2r$ with the $i$-th one corresponding to edge $e_i$ of $G$. Intuitively, each
of the sections is dedicated to checking that we did not choose two vertices
incident to the same edge. For each original node $i$ and each edge $e_j$ of
$s$, let $i_{e_j}$ be the $(2jr)$-th ancestor of $G$. If $i'$ is the parent of
$i_{e_j}$, we set $\delta(i_{e_j}i')=3p(i_{e_j}i')+1$. For vertices $i$ that
are not one of the $i_{e_j}$, we set $\delta(i_{e_j}i')=3p(i_{e_j}i')+2$. Now,
for any node $i$ in $T'$ and its parent $i'$ in $T'$, we add edges between
$a_{i,j}$ and $\{a_{i',\ell}: \ell \in [1,\min(L',L-L'-\delta(ii')+j)]\}$.

\subparagraph{Clique chains.} For each original node $i$, and each color $c \in
[1,k]$, we add a clique chain $C_{i,c}$. Each clique chain consists of
cliques $C_{i,c,\gamma}$ with $\gamma \in [1,2N+r+1]$. There is a cluster
gadget on each $C_{i,c,\gamma}$. $C_{i,c,1}$ is fully adjacent to $A_i$,
$C_{i,c,2N+r+3}$ is fully adjacent to $A_{i'}$, where $i'$ is the original
grandparent of $i$, and $C_{i,c,\gamma}$ is fully adjacent to
$C_{i,c,\gamma+1}$ for $\gamma \in [1,2N+r]$. The clique $C_{i,c,\gamma}$ has
size $4$ if there is an edge $e_j \in E$ with endpoint $v_{i,c,\alpha}$, and,
either $\gamma=2jr+\alpha$ or $\gamma=2jr+\alpha+N+1$, otherwise, it has size $3$.

\subparagraph{Rooting gadget.} This gadget ensures a rooting of the tree-layout
with $A_r$ with $r$ the root of $T'$, and such that $A_r$ is also correctly
ordered, we add $L+1$ copies of a clique $R=\{r_1,\dots,r_L\}$ of size $L$
such that $a_{r,j}$ is adjacent to exactly $r_1,\dots,r_{L-j+1}$ in each copy.

\begin{claim}
	If $I$ is a valid independent set of our instance, we deduce a graph $H$ with treebandwidth at most $L$.
\end{claim}

\begin{subproof}
	We construct a tree-layout of $H$. We start from $T'$, and replace each
	node $i$ by $a_{i,L'},\dots,a_{i,1}$ in order from closest to the root to
	furthest to the root. For each original node $i$, let $h(i,c)$
	be the index of vertex of $V_{i,c}$ in the independent set, i.e.
	$v_{i,c,h(i,c)} \in I$. We place $C_{i,c,h(i,c)+j}$ above $A_{i'}$ for $i'$
	the $j$-th ancestor of $i$ in $T'$, with $j \in [0,2N+1]$. The rest of the
	clique chain can be folded greedily (see Figure~\ref{fig:subdivision}) in a
	branch below $A_i$ and in a branch below $A_{i'}$, with $i'$ the original
	grandparent of $i$.

	There are exactly $p(ii')$ cliques from the clique chains between $A_i$ and
	$A_{i'}$ with $i'$ the parent of $i$. By construction, there are at most
	two cliques of size $4$ between $A_i$ and $A_{i'}$. Moreover, since $I$ is
	an independent set, for $i=i_{e_j}$ and $i'$ the parent of $i$,
	there is at most one clique of size $4$ between $A_i$ and $A_{i'}$.
	By definition of $\delta(ii')$, the distance in the layout for edges
	between $A_i$ and $A_{i'}$ is at most $L$. Edges between consecutive
	cliques in the clique chains have distance at most $L'+3k+4 \leq L$, if the
	relative order of the chains is the same between each pair $A_i,A_{i'}$.
	Edges in folded parts are also shorter than $L'+3k+4 \leq L$.
\end{subproof}

\begin{claim}
	If $H$ has treebandwidth at most $L$, then the instance of TCMIS is positive.
\end{claim}

\begin{subproof}
	We first observe that the trunk cannot be folded in a tree-layout of
	treebandwidth at most $L$. Indeed, each $A_i$ must appear consecutively due
	to its cluster gadget. If $i'$ is the parent of $i$ in $T'$, then all
	vertices of $A_{i'}$ must be comparable to all vertices of $A_i$. 
	Furthermore, we can fit exactly $p(ii')$ cliques between $A_i$ and
	$A_{i'}$: we cannot fit more since $3(p(ii')+1)>3p(ii')+2$, and
	if we skip a clique, then there will be edges covering a distance
	at least $2L'+5 > L$. We recall that cliques in the clique chains have
	clustering gadgets and appear consecutively as well. Furthermore, one can
	check that $A_r$ is enforced as the root thanks to the rooting gadget. We
	now consider a tree-layout of treebandwidth at most $L$ for $H$. From the
	observations above, it follows that the $A_i$s follow the structure of
	$T'$, and that there is exactly one clique from each clique chain between a
	pair of consecutive $A_i$s. Let $h(i,c)$ be the minimum integer such that
	$C_{i,c,h(i,c)}$ appears above $A_i$, i.e. it is first clique of a clique
	chain to appear on the trunk. Observe first that $h(i,c) \in [1,r]$ because
	there must be $2(N+1)$ cliques from the chain placed on the trunk so the
	number of cliques folded below $A_i$ is at most $r-1$. We claim that
	picking $v_{i,c,h(i,c)}$ for each original node $i$, and colour $c \in
	[1,k]$, defines an independent set. Indeed, for each pair of distinct
	vertices $v_{i,c,h(i,c)},v_{i',c',h(i',c')}$ of $G$, they can be adjacent
	only if $i=i'$ or (w.l.o.g.) $i'$ is the parent of $i$ in $T$. If they
	were, there would be an edge $e_j$ incident to both, in which case the
	cliques $C_{i,c,\gamma},C_{i',c',\gamma'}$ for $\gamma=h(i,c)+2jr \mod
	(N+1)$ and $\gamma'=h(i',c')+2jr \mod (N+1)$ would have size $4$. In
	particular, this means that the total size of cliques from clique chains
	directly above $A_{i_{e_j}}$ is more than $\delta(i_{e_j}i')$ with $i'$ the
	parent of $i_{e_j}$. This would contradict the bound on the distances
	between adjacent vertices of $A_{i'}$ and $A_{i_{e_j}}$.
\end{subproof}
\vspace{-10pt}
\end{proof}

\begin{theorem}
\textsc{Treebandwidth} is in XALP, and is XALP-hard even when parameterised by
$\tw + \Delta$.
\end{theorem}

\section{Parameters and topological minor obstructions}\label{sec:zoo}

In order to present how the parameters we introduce compare to other parameters from the
literature, we give a list of parameters that can be characterised by topological minors.
The relations between the considered parameters are illustrated in
Figure~\ref{fig:hasse-param}. Most characterisations are known or simple consequences of
existing literature.
Our results are the following: a proof of a characterisation by topological
minors for tree-cut width, a proof that treespan and domino treewidth are
equivalent, and an improved bound on edge-treewidth.

The obstruction lists corresponding to the considered parameters are given in Table~\ref{tab:obs-lists} with
some obstructions illustrated in Figure~\ref{fig:obs-graphs}. The list is small in all
considered examples. However, considering a minimal list may lead to large gaps between
obstruction size and value of the parameter. As a simple example, consider component size
which is exactly the size of the largest tree subgraph, but there is an exponential gap
with maximum degree and the longest path which asymptotically characterise the size of the
largest tree.

Characterising a parameter $\kappa$ by its minimal obstructions makes it simple to compare
to other parameters, to estimate the value of $\kappa$ on structured instances (e.g.
gadgets in hardness proofs), and may allow for win-win approaches.

\begin{figure}[h]
\centering
\begin{tikzpicture}
\node[draw,rectangle] (sz) at (0,0) {size};
\node[draw,rectangle] (vc) at (-2,-1) {vertex cover};
\node[draw,rectangle] (csz) at (2,-1) {component size};
\node[draw,rectangle] (td) at (0,-2) {treedepth};
\node[draw,rectangle] (pw) at (-1.5,-5) {pathwidth};
\node[draw,rectangle] (cw) at (2.5,-3) {cutwidth};
\node[draw,rectangle] (fvs) at (-3,-6) {feedback vertex set};
\node[draw,rectangle] (dtw) at (3,-4) {domino treewidth};
\node[draw,rectangle] (stcw) at (2,-5) {slim tree-cut width};
\node[draw,rectangle] (etw) at (3,-6) {edge-treewidth};
\node[draw,rectangle] (tcw) at (1.5,-7) {tree-cut width};
\node[draw,rectangle] (delta) at (6,-6) {maximum degree};
\node[draw,rectangle] (delta2) at (7,-8) {biconnected maximum degree};
\node[draw,rectangle] (tpw) at (2,-8) {tree-partition-width};
\node[draw,rectangle,dashed] (tbw) at (1.2,-9) {treebandwidth};
\node[draw,rectangle] (tw) at (0,-10) {treewidth};
\node[draw,rectangle,dashed] (otw) at (4.5,-9) {overlap treewidth};
\node[draw,rectangle] (ntd) at (3.5,-10) {fan number};
\node[draw,rectangle] (dipole) at (7,-10) {dipole number};

\draw (sz) -- (vc) -- (fvs) -- (tw);
\draw (sz) -- (csz) -- (td) -- (vc);
\draw (csz) -- (cw) -- (pw) -- (td) to [out=270,in=150] (tbw) -- (tw) -- (pw);
\draw (cw) -- (dtw) -- (stcw) to [out=220,in=110] (tcw) -- (tpw) -- (tbw) -- (ntd) to
[out=10,in=250] (delta2) -- (dipole) -- (otw) -- (etw) -- (delta2) -- (delta) -- (dtw);
\draw (stcw) --  (etw) to [out=270,in=30] (tpw);
\draw (otw) to [out=200,in=20] (tw);

\end{tikzpicture}
\caption{Hasse diagram of some parameters that admit a simple characterisation by
topological minors. Parameters we introduce are in dashed rectangles.}
\label{fig:hasse-param}
\end{figure}
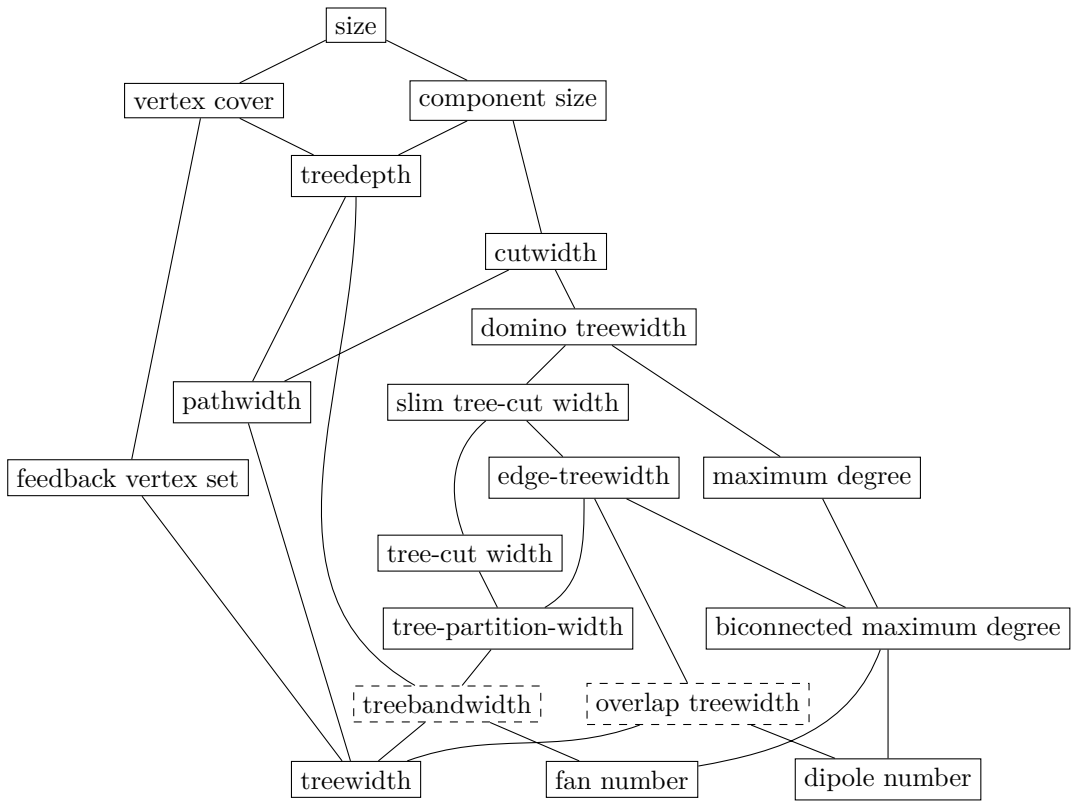

\begin{figure}[h]
\centering
\includegraphics{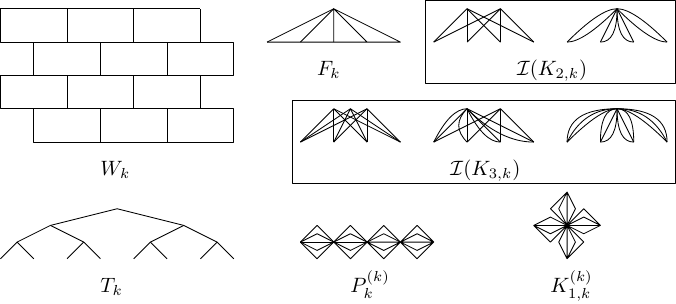}
\caption{Some graphs appearing in obstruction lists. $\mathcal{I}(K_{c,k})$ are the graphs
that can be obtained from $K_{c,k}$ by identifications on its set of vertices of degree $k$.}
\label{fig:obs-graphs}
\end{figure}

\begin{table}[h]
\centering
\begin{tabular}{|c|c|}
\hline
Parameters & Obstructions  \\
\hline
size & $kK_1$ \\
vertex cover & $kK_2$ \\
feedback vertex set & $kK_3$ \\
component size & $K_{1,k},P_k$ \\
treedepth & $P_k$ \\
pathwidth & $T_k$ \\
cutwidth & $T_k,K_{1,k}$ \\
treewidth & $W_k$ \\
domino treewidth & $W_k,K_{1,k}$ \\
slim tree-cut width & $W_k,\mathcal{I}(K_{2,k})$ \\
edge-treewidth & $W_k,F_k,K_{2,k}$ \\
maximum degree & $K_{1,k}$ \\
tree-cut width & $W_k,F_k,P_k^{(k)},\mathcal{I}(K_{3,k})$ \\
tree-partition-width & $W_k,F_k,P_k^{(k)},K_{1,k}^{(k)}$ \\
biconnected maximum degree & $F_k,K_{2,k}$ \\
treebandwidth & $W_k,F_k$ \\
overlap treewidth & $W_k,K_{2,k}$ \\
fan number & $F_k$ \\
dipole number & $K_{2,k}$ \\
\hline
\end{tabular}
\caption{Lists of parametric obstructions by topological minors.}
\label{tab:obs-lists}
\end{table}

A very natural way of restricting a graph class is by imposing a bound on the
maximum degree. This is equivalent to bounding the largest star appearing as a
subgraph. In this section, we illustrate how parameters originating from degree bounds,
forbidden minors, and forbidden immersions, can be characterised by simple
classes of forbidden topological minors.

Contrary to the minor relation and the immersion relation, which
were shown to be well-quasi-orderings by Robertson and Seymour in the Graph
Minor series \cite{RobertsonS04,RobertsonS10}, the topological minor relation
is known not to be a well-quasi-ordering.
Nonetheless, Robertson conjectured that there is only one minimal antichain
based on $2$-multiples of paths. A proof of this conjecture was given by
Chun-Hung Liu in his PhD thesis \cite{ChunHungLiuThesis}. 

Questions about the structure of graphs excluding a (topological) minor are naturally
related to questions about unavoidable (topological) minors. A class $\mathcal{U}$ of
unavoidable (topological) minors in a host class $\mathcal{C}$ of graphs consists of
graphs such that any large enough graph of $\mathcal{C}$ contains a graph of $\mathcal{U}$
as a (topological) minor.
In \cite{typicalsubgraphs}, the unavoidable topological minors of $c$-connected
graphs for $c \leq 4$ are given. This was extended by Benson Joeris in his PhD
thesis \cite{BensonJoerisThesis,generalizedwheel} who describes large
\emph{generalized wheels} as the unavoidable minors of graphs that
contain large $c$-connected sets of vertices (now for any integer $c$).
Generalized wheels (very) roughly consist of circular grids with an apex set
hitting each column of the grid. Although this is not explicitly stated, a
close inspection of the proof reveals that the apex set is not obtained by
contractions. In particular, if we restrict the grid part to be subcubic, we
obtain a topological minor. In the setting of infinite graphs, Carolyn Chun and
Guoli Ding obtained a similar result for topological minors
\cite{UnavoidableTopMinor}. 

We now list the results that are summarised in Table~\ref{tab:obs-lists}.

The result below follows from picking all vertices of a maximal matching and was
discovered independently by Fanica Gavril and Mihalis Yannakakis.

\begin{theorem}
	In a graph $G$, we have $\mu(G) \leq \vc(G) \leq 2 \mu(G),$ where $\mu(G)$ is
	the maximum matching number of $G$.
\end{theorem}

The following theorem due to Paul Erd\H{o}s and Lajos P\'{o}sa gives its name to the
Erd\H{o}s-P\'{o}sa property, a duality between packing subgraphs and hitting subgraphs.

\begin{theorem}[Erd\H{o}s and P\'{o}sa \cite{ErdosPosa}]
	In a graph $G$, we have $\nu(G) \leq \fvs(G) \leq O(\nu(G) \log \nu(G)),$
	where $\nu(G)$ is the cycle packing number of $G$.
\end{theorem}

The theorems above imply that the parameters $\vc$ and $\fvs$ are tied to the parametric
graphs $kK_2$ and $kK_3$, respectively.

The obstructions to treedepth, pathwidth, and treewidth are well-known, they are the path, the
complete binary tree, and the grid (or the wall), respectively. See \cite{ChuzhoyTan21} for
the best known bound between treewidth and the wall number. See \cite{CzerwinskiNP21} for
polynomial bounds for treedepth and consequently pathwidth in terms of the three obstructions.

\begin{theorem}
	Treedepth and pathwidth are tied to the size of the largest path and size of the
	largest complete binary tree topological minor.

	Treewidth is polynomially tied to the wall number.
\end{theorem}

The best known bounds on domino treewidth compared to treewidth and maximum degree are the
following. The upper bound is not known to be tight, in particular it might be possible to
shave off a $\Delta$ factor. See \cite{DominoTw2,tpwWood} for constructions such that $\dtw(G)
= \Omega(\Delta(G) \tw(G))$.

\begin{theorem}[Bodlaender \cite{DominoTw2}]\label{thm:domino-tw}
	In a graph $G$, we have $$\max(\tw(G),(\Delta(G)-1)/2) \leq \dtw(G) \leq
	O(\tw(G)\Delta(G)^2).$$
\end{theorem}

We show that the parameter treespan ($\ts(G)$) introduced in
\cite{FominHT05GraphSearching} roughly corresponds to the same notion as domino
treewidth. One of the equivalent definitions introduced as \emph{elimination
span} in \cite{FominHT05GraphSearching} can be expressed in terms of
tree-layout. The pruned subtree of a vertex $x$ in a tree-layout $T$ is the
subtree $T^p[x]$ obtained from the subtree of $T$ rooted in $x$ by removing all
vertices that do not have a neighbour of $x$ in their descendants. The
elimination span of a tree-layout is then maximum number of vertices in a
pruned subtree minus one, and the elimination span of a graph is the minimum elimination
span of a tree-layout of the graph.

\begin{theorem}\label{thm:eq-treespan-dtw}
Treespan is polynomially tied to domino treewidth.
\end{theorem}

\begin{proof}

The statement follows from a lower bound by $\tw(G)$ and $\Delta(G)$, from
which we deduce a lower bound in terms of domino treewidth using
\cref{thm:domino-tw}, and a direct upper bound in terms of domino treewidth.

\begin{claim}
$\max(\tw(G),\Delta(G)/2) \leq \ts(G)$
\end{claim}

\begin{subproof}
Consider a fixed tree-layout $T$ of $G$.

Let $x$ be a vertex of degree $\Delta(G)$.
Let $d_1(x)$ (resp. $d_2(x)$) be the number of neighbours of $x$ in $G$ that
are ancestors (resp. descendants) of $x$ in $T$. Since $d_1(x)+d_2(x) =
\Delta(G)$, one of $d_1(x)$ and $d_2(x)$ is at least $\Delta(G)/2$. 
First, we consider the case when $d_1(x)\geq \Delta(G)/2$. Let $y$ be
the neighbour of $x$ in $G$ which is closest to the root in $T$. Then $x \in
T^p[y]$ and $|T^p[y]| \geq d_1(x)+1$, thus, $\ts(G) \geq \Delta(G)/2$.  We now
move to the case when $d_2(x)\geq \Delta(G)/2$. Then, $|T^p[x]| \geq d_2(x) +
1$, thus $\ts(G) \geq \Delta(G)/2$.

There must exist a vertex $x$ such that the vertices of $G$ mapped to the
subtree $T[x]$ of $T$ rooted in $x$ have a set $U$ of neighbours in $G$
that are ancestors of $x$ such that $|U|\geq \tw(G)$. (This follows from the
characterisation of treewidth via the maximum clique size in a minimal chordal
completion.) Let $y$ be the vertex of $U$ which is closest to the root in $T$.
Then, $T^p[y]$ contains all vertices of $U \cup \{x\}$ since $y$ has a
neighbour which is a descendant of $x$ and vertices of $U$ are all on the path
of $T$ from $y$ to $x$. We conclude that $|T^p[y]|\geq \tw(G)+1$, implying
$\ts(G) \geq \tw(G)$.
\end{subproof}

\begin{claim}
$\ts(G) \leq 2\dtw(G)$
\end{claim}

\begin{subproof}
We now consider the canonical tree-layout $T'$ of $G$ obtained from a domino
tree decomposition $(T,X)$ of width $k$. The pruned subtree of a vertex $x$ is
always contained in the path of $T'$ induced by the two bags (of size at most
$k+1$) of $T$ containing $x$. Since the bags have at least one vertex in
common, we conclude that $\ts(G) \leq 2k$.
\end{subproof}
\vspace{-12pt}
\end{proof}

\begin{lemma}[see e.g. \cite{PhucNguyenThesis}]\label{lem:ramsey-connected}
If a graph $H$ is connected and $X \subseteq V(H)$ is large enough (i.e. $|X|
\geq f(s)$), then either $H$ contains a subdivided $K_{1,s}$ with leaves in
$X$, or it contains an $X$-rooted path minor of length $s$.
\end{lemma}

By applying \cref{lem:ramsey-connected} with the full vertex set of a connected graph, we
have the characterisation of component size\footnote{This is also easily derived from the
fact a connected graph admits a spanning tree and the size of a tree is bounded by its
maximum degree and depth.} by $P_k$ and $K_{1,k}$.

By considering the vertex $v$ of maximum degree in a $2$-connected component
$C$ and setting $H:=C-v$, $X:=N_C[v]$, we may derive the following.

\begin{observation}
$G$ has bounded degree in each $2$-connected component if and only if it contains neither
a large subdivided fan nor a large subdivided dipole.
\end{observation}

In particular this implies that the parameter edge-treewidth introduced in
\cite{edgetreewidth} may be asymptotically characterised by topological minors
and does not strictly require the use of weak topological minors. Let $\Delta_2(G)$ denote
biconnected maximum degree of $G$, which is the maximum over $2$-connected components $C$
of $G$ of $\Delta(C)$. It was shown in \cite{edgetreewidth} that edge-treewidth is
functionally equivalent to $\tw(G)+\Delta_2(G)$. We give an improved upper bound and
obtain the following bounds.

\begin{theorem}
	In a graph $G$, we have $\max(\tw,\sqrt{\Delta_2}) \leq \etw \leq \Delta_2\tw$. 

	Edge-treewidth is polynomially tied to the largest of the treewidth and the
	biconnected maximum degree.
\end{theorem}

We use the following definition of edge-treewidth which mimics cutwidth: given a
tree-layout $T$ of $G$, its edge-treewidth is the maximum over nodes $u$ of $T$ of the
number of edges of $G$ with one endpoint mapped to a descendant of $u$ (including $u$) and
one endpoint mapped to a strict ancestor of $u$. Similarly, treewidth can be characterised
by the maximum over nodes $u$ of the number of neighbours of the descendants $u$ that are
strict ancestors of $u$.

\begin{proof}
The lower bound is proven in \cite{edgetreewidth}, we prove the upper bound.

Observe that given a graph $G$ of treewidth $w$, we can obtain a tree-layout of treewidth
$w$ rooted at any vertex of $G$. It follows from the definition that such a tree-layout
has edge-treewidth is at most $w\Delta(G)$.

We now describe a construction of tree-layout of $G$ of edge-treewidth
$\Delta_2(G)\tw(G)$.

First, compute the biconnected components of $G$ (this can be done in linear time). 
Then, root the tree of biconnected components and, for each biconnected component, if
defined, produce a tree-layout of optimal treewidth rooted at the cutvertex shared with the parent biconnected component. Let $T$ be the obtained tree-layout.

For each edge of $T$, only one biconnected component of $G$ has edges crossing it. We
conclude that the edge-treewidth of $T$ is bounded by the bound for the tree-layouts of
its biconnected components. This shows $\etw(G) \leq \Delta_2(G)\tw(G)$.
\end{proof}

The topological minor obstructions to tree-partition-width were given by Ding and
Oporowski \cite{obs-tpw}.

\begin{theorem}\label{thm:obs-tpw}
	The tree-partition-width of a graph $G$ is tied to the largest $k$ such that one of
	$$W_k,F_k,K_{1,k}^{(k)},P_k^{(k)}$$ is a topological minor of $G$.
\end{theorem}

To give the list of topological minor obstructions to tree-cut width, we use the following
theorem of Wollan \cite{treecutwidth}.

\begin{theorem}
	The tree-cut width of a graph $G$ is tied to the largest $k$ such that $W_k$ is a weak
	immersion of $G$. 
\end{theorem}

We denote by $\Theta'_k$ the graph consisting of an independent set $I$ of size $k$
to which we add two vertices $x,y$, two parallel edges $xu$ for each $u\in I$, and an edge $yu$ for each $u \in I$.
The graphs $\Theta'_k,K_{1,k}^{(3)}$ and $K_{3,k}$ correspond to the graphs
$\mathcal{I}(K_{3,k})$ obtained by identifications of the vertices of degree
$k$ in $K_{3,k}$. 

\begin{theorem}
	The tree-cut width of a graph $G$ is tied to the largest $k$ such that one of
	$$W_k,F_k,K_{1,k}^{(3)},\Theta'_k,K_{3,k},P^{(k)}_k$$ is a topological minor
	of $G$.
\end{theorem}

\begin{proof}
	First, we observe that all given parametric obstructions contain large wall immersions.
	Indeed, $P^{(k)}_k$ contains every graph with $k$ edges as a weak
	immersion\footnote{Note however that it does not contain large
	$3$-connected strong immersions.},
	$K_{1,k}^{(3)},\Theta'_k,$ and $K_{3,k}$ contain every subcubic graph on
	$k$ vertices as an immersion. To find a wall immersion in $F_k$, divide its path
	into $\sqrt{k}$ subpaths of length $\sqrt{k}$, they are the rows of our
	wall and we may route edges between distinct rows of our wall via the
	universal vertex. We conclude that all our obstructions of order $k$
	contain a wall of order $\Omega(\sqrt{k})$ as an immersion.

	We now show that if the tree-cut width is not bounded by a function of $k$, then
	we find a parametric obstruction of order at least $k$.
	
	If $G$ has unbounded tree-partition-width, it contains a subdivision of
	$W_k,F_k,K_{1,k}^{(k)},$ or $P_k^{(k)}$ (see \cref{thm:obs-tpw}). Since
	$K_{1,k}^{(k)}$ contains $K_{1,k}^{(3)}$, in all cases, we find an obstruction from
	our list of order at least $k$.

	We now assume that $G$ has tree-partition-width bounded by a function of
	$k$, and since $G$ has unbounded tree-cut width, it contains a large wall
	immersion. We will analyse how the immersion model of a large wall can fit
	in a tree-partition of bounded width. A bag of the tree-partition contains
	a bounded number of vertices and hence a bounded number of branch vertices.
	Every edge $tt'$ of the tree-partition corresponds to a cut of $G$ crossed
	by a number of edges bounded by a function of $k$. Walls being well-linked,
	one side of this cut contains a number of branch vertices bounded by a
	function of $k$. We orient edges of the tree-partition towards the subtree
	that contains most branch vertices in its bags. This orientation has a
	unique sink, and due to the previous observations, it is a bag of degree
	unbounded by a function of $k$. More precisely, the number of edges of the
	tree-partition incident to this bag that lead to branch vertices of degree
	$3$ in the model is unbounded by a function of $k$. In particular, for each
	of these edges, we associate the multiset corresponding to the first
	vertices of the central bag on these 3 paths, with multiplicity. Since the
	central bag contains a number of vertices that is bounded by a function of
	$k$, the number of multisets of order $3$ on this ground set is also bounded
	by a function of $k$. In particular, one multiset is repeated a number of
	times unbounded by a function of $k$ among our edges of the tree-partition.
	If this multiset has a vertex of multiplicity $3$, we find a subdivision of
	$K_{1,k}^{(3)}$. Else, if it has a vertex of multiplicity $2$, we find a
	subdivision of $\Theta'_k$. Else, we find a subdivision of $K_{3,k}$.
	This requires the rather technical observation that, from an edge-disjoint
	tripod, we can always extract a vertex-disjoint tripod with the same
	leaves.
\end{proof}

The obstructions for slim tree-cut width were given in \cite{stcw}.

\begin{theorem}
	The slim tree-cut width of a graph $G$ is tied to the largest $k$ such that
	one of $$W_k,K_{2,k},K_{1,k}^{(2)}$$ is a topological minor of $G$.
\end{theorem}

\section{Equivalent definitions}\label{sec:def}

We use the following characterisation of proper chordal graphs \cite{properchordal}.

\begin{definition}\label{def:proper-chordal}
A graph $G$ is a proper chordal graph if there exists a tree-layout $T$ of $G$ such that
for every maximal clique $K$ of $G$, the vertices of $K$ appear consecutively on a
root-to-leaf path of $T$.
\end{definition}

\begin{proposition}
Given a fixed graph $G$, the following statements are equivalent:
\begin{enumerate}[label=(\roman*)]
\item\label{enum:tbw} $\tbw(G) \leq k$;
\item\label{enum:proper-chordal} $\omega(H)-1 \leq k$ for a proper chordal supergraph $H$
of $G$; and
\item\label{enum:search} there is a monotone search strategy to capture a visible fugitive
of infinite speed in which searchers are placed one at a time and each vertex of $G$ is
occupied by a searcher during at most $k+1$ time steps.
\end{enumerate}
\end{proposition}

\begin{proof}
\begin{description}[listparindent=1.5em]
\setlength{\itemsep}{7pt}

\item[\ref{enum:tbw} $\Rightarrow$ \ref{enum:proper-chordal}] Consider a tree-layout $T$
of bandwidth at most $k$. Define $H$ as the graph obtained from $G$ by adding edges to
form cliques on intervals $[x,y]$ of $T$, for each edge $x,y$ of $G$. By interval $[a,b]$
of $T$, we mean the set of vertices of $G$ that are on the unique path from $a$ to $b$ in
$T$. We show that $H$ is a proper chordal graph using \cref{def:proper-chordal}. Let $K$
be a maximal clique of $H$ and let $x,y$ be its vertices that are closest and furthest,
respectively, to the root of $T$. By construction of $H$ and definition of $K$, $xy \in
E(G)$. Indeed, if $xy \notin E(G)$, then $K$ is strictly contained in an interval $[a,b]$
of $T$ for which $ab \in E(G)$. As $[a,b]$ would induce a clique in $H$, this would
contradict the maximality of $K$. We conclude that vertices of a maximal clique of $H$ are
intervals of $T$, i.e. they are consecutive on a root-to-leaf path of $T$. We conclude
that $H$ is proper chordal, but also that the size of a maximum clique in $H$ corresponds to
the size of an interval $[x,y]$ of $T$ such that $xy \in E(G)$. In turn, we deduce that
$\omega(H)-1 \leq k$ since the distance between $x$ and $y$ is bounded by $k$.

\item[\ref{enum:proper-chordal} $\Rightarrow$ \ref{enum:tbw}] If $H$ is a proper chordal
graph with $\omega(H)-1 \leq k$, then, by \cref{def:proper-chordal}, $\tbw(H) \leq k$. We
deduce $\tbw(G) \leq \tbw(H) \leq k$ since $G$ is a subgraph of $H$. 

\item[\ref{enum:tbw} $\Rightarrow$ \ref{enum:search}] Given a tree-layout $T$ of bandwidth
at most $k$, we describe a winning strategy for the searchers. At each time step, place a
searcher at the root of the subtree $T_i$ containing the fugitive. We may remove a
searcher from the vertex he occupies when it is not in $N_G(V(T_i))$. Indeed, by the monotonicity
assumption, its neighbours will never become accessible to the fugitive again. In
particular, any vertex at distance at least $k$ from the root of $T_i$ may be removed. Due
to the order in which we place searchers, a searcher on the $k$-th ancestor $x$ of the
root of $T_i$ is removed $k$ steps after it was placed. This means $x$ was occupied for
$k+1$ time steps. Since $T$ has finite depth, the strategy eventually leads to the capture
of the fugitive.

\item[\ref{enum:search} $\Rightarrow$ \ref{enum:tbw}] Given a monotone strategy for the
searchers that occupies a vertex during at most $k+1$ time steps, we construct a
tree-layout $T$ of bandwidth at most $k$. Consider $H$ a history of moves of the searchers
and $S$ inducing a connected subgraph such that there exists an execution with history $H$ where $S$ is the set of vertices accessible to the fugitive. We define $T_{S,H}$ on such a pair inductively as follows. If $S$ is empty, $T_{S,H}$ is an empty tree. Otherwise, let $x$ be the next move of the strategy in configuration $(S,H)$ and let $S_1,\dots,S_\ell$ denote the connected components of $G[S] - x$. After the searcher
is placed on $x$, the fugitive must be in some $S_i$. For each $S_i$, if the fugitive
chose to move to $S_i$, the strategy can capture the fugitive such that no vertex will
have been occupied more than $k+1$ steps. Let $H'$ be obtained from $H$ by adding $x$ as
the last move. We set $x$ to be the root of $T_{S,H}$ with subtrees $T_{S_1,H'},\dots,
T_{S_\ell,H'}$. Consider $T=T_{V(G),H}$ with $H$ an empty history of moves for the
searchers. For each rooted subtree of $T_i=T_{S,H}$, by monotonicity, there were still
searchers on $A=N(V(T_i))$ when placing a searcher on the root of $T_i$. By definition
of the strategy, a vertex $a$ of $A$ cannot be occupied more than $k+1$ steps.
Furthermore, by definition of $A$, there is an edge $ay$ between $a$ and $S$. By
monotonicity of the strategy, the sets $S'$ of vertices accessible to the fugitive after $a$
was placed always contained $S$, and thus $a$ was part of $N(S')$. Hence, at each time
step since we placed a searcher on $a$, $a$ was occupied. We conclude that $a$ cannot be
at distance more than $k$ in $T$ of the root of $T_i$, and $T$ is the desired tree-layout.
\qedhere
\end{description}
\end{proof}

\begin{opquestion}
Is there always an optimal monotone strategy in this graph searching game?
\end{opquestion}

\bibliography{refs}

\end{document}